\documentclass[12pt, draftclsnofoot,onecolumn]{IEEEtran}
\usepackage{graphicx}
\usepackage{cite}
\usepackage{enumerate}
\usepackage{array}
\usepackage{amsmath}
\usepackage{amssymb}
\usepackage{subeqnarray}
\usepackage{cases}
\usepackage{bm}
\usepackage[noend]{algpseudocode}
\usepackage[nothing]{algorithm}
\usepackage{multirow}
\usepackage{epsfig}
\usepackage{epstopdf}
\usepackage{xcolor}
\usepackage{subfigure}
\newtheorem{definition}{Definition}
\newtheorem{theorem}{Theorem}
\newtheorem{lemma}{Lemma}
\newtheorem{proposition}{Proposition}

\newenvironment{proof}{{\bf Proof:}}{\hfill\rule{2mm}{2mm}}
\newcommand{\Prob}[1]{\mathrm{P}\left\{#1\right\}}
\newcommand{\E}[1]{\mathrm{E}\left[#1\right]}

\begin{document}
\title{Urgency of Information for \\Context-Aware Timely Status Updates \\in Remote Control Systems}

\author{Xi~Zheng, Sheng~Zhou,~\IEEEmembership{Member,~IEEE,} Zhisheng~Niu,~\IEEEmembership{Fellow,~IEEE}
\thanks{Part of the paper has been presented in IEEE GLOBECOM'19 \cite{previous} and has been accepted by 6G Summit \cite{previous_2}. (corresponding author: Sheng Zhou)

The authors are with Beijing National Research Center for Information Science and Technology, Department of Electronic Engineering, Tsinghua University, Beijing 100084, China. Emails: zhengx14@mails.tsinghua.edu.cn,
\{sheng.zhou, niuzhs\}@tsinghua.edu.cn.}}

\maketitle

\begin{abstract}
As 5G and Internet-of-Things (IoT) are deeply integrated into vertical industries such as autonomous driving and industrial robotics, timely status update is crucial for remote monitoring and control. In this regard, Age of Information (AoI) has been proposed to measure the freshness of status updates. However, it is just a metric changing linearly with time and irrelevant of context-awareness.  We propose a context-based metric, named as Urgency of Information (UoI), to measure the nonlinear time-varying importance and the non-uniform context-dependence of the status information. This paper first establishes a theoretical framework for UoI characterization and then provides UoI-optimal status updating and user scheduling schemes in both single-terminal and multi-terminal cases. Specifically, an update-index-based scheme is proposed for a single-terminal system, where the terminal always updates and transmits when its update index is larger than a threshold. For the multi-terminal case, the UoI of the proposed scheduling scheme is proven to be upper-bounded and its decentralized implementation by Carrier Sensing Multiple Access with Collision Avoidance (CSMA/CA) is also provided. In the simulations, the proposed updating and scheduling schemes notably outperform the existing ones such as round robin and AoI-optimal schemes in terms of UoI, error-bound violation and control system stability.
\end{abstract}


\section{Introduction}
Emerging applications in 5G and Internet of Things (IoT) require timely and reliable status updates to enable remote monitoring and control\cite{network}\cite{tactile}. Taking autonomous driving as an example, vehicles need to timely exchange position, velocity, acceleration and driving intention information (called ``status information'' in the sequel) to enable driving assistance applications such as collision avoidance, intersection crossing, and platoon driving.  Here, note that there are two major differences from the traditional networks.  One is that the status information has a unique feature, called Markovian feature, i.e., the old (or existing) status packets can be completely replaced by the newly received (or updated) status information or, alternatively speaking, the system performance relies only on the latest status information.  Another is that the timeliness requirement is different from the traditional communication latency requirement, where the timeliness is counted from the generation epoch of the status information and therefore
consists of the delay until being updated (or sampled), the delay until being scheduled, and the communication latency through the network. We will call the timeliness requirement as \emph{information latency} in the sequel, which can be considered as the extension of the communication latency.

To characterize the information latency, a new metric called \emph{Age of Information} (AoI) has been introduced \cite{aoi}, which is defined as the time elapsed since the generation epoch of the most up-to-date packet received.  This is indeed a very good metric to measure the freshness of the status information. Generally speaking, the smaller the AoI, the fresher the status information and therefore the less the information latency.  As a result, there have been a lot of efforts devoted to the researches on AoI \cite{aoi}--\!\!\cite{NegExpo}. Some work focuses on the analysis of average AoI by considering periodic or random sampling and transmission \cite{6195689}\cite{DBLP:journals/corr/abs-1805-12586}\cite{SamplePath}, energy harvesting terminals \cite{AoIEH}\cite{Arafa18}, multiple transmission paths \cite{7364263} and packet blockage \cite{7415972}, while others investigate sampling and transmission scheduling for multi-access networks \cite{8424039}\cite{DBLP:journals/corr/abs-1801-01803}--\!\!\cite{Jiang}, energy harvesting sources \cite{7308962}--\!\!\cite{Arafa18}, and systems with random delay \cite{8000687}\cite{7541763}. In particular, ref. \cite{NonlinearFunc}--\!\!\cite{8424039} characterize the non-linear loss caused by information staleness as a function (e.g., exponential, logarithmic, and step function) of AoI, which adds non-linearity to the analysis of status information. However, by investigating the mean square error (MSE) minimization in the remote estimation of Wiener process \cite{wiener} and Ornstein-Uhlenbeck process \cite{NegExpo}, it is proved that AoI-based sampling and transmission is not optimal when status information is observable. In \cite{EffectiveAge}, the authors claim the necessity for ``effective'' age metrics that are minimized when MSE is minimized, and consider MSE minimization for the remote estimation of two-state Markov signals. In addition, authors of \cite{AoS} propose Age of Synchronization (AoS) for applications like cache design to measure the length of time during which the cache is not synchronized. Similarly, Age of Incorrect Information (AoII), defined as the time elapsed since the last time that the status information is correct, is proposed as a timeliness metric in \cite{AoII}. However, none of these metrics takes context information into consideration.

In practical systems, the impact or the importance of the information latency on the system performance is inevitably nonlinear as well as context-dependent.  Generally speaking, when the status changes rapidly, more frequent information updates are required. Timeliness requirement is also related to context information \cite{context}\cite{context-survey}, which includes all the knowledge about the underlying system that determines how important the status is. For example, when the system is at a critical situation (e.g., approaching to an intersection or overtaking the front vehicles), its status should be more frequently updated. Otherwise, insufficient status updates will hinder the effectiveness of the applications, resulting in unacceptable performance degradation. On the other hand, excessive status information deliver will bring little marginal performance gain while consuming extra wireless and energy resources. Therefore, to ensure the timeliness of information delivery and the effectiveness of status-based control, status updates should adapt to the context of the system and the non-uniform evolution of status.

In this regard, the AoI metric is not perfect for the characterization of information latency (i.e., timeliness of the status information) because it is in nature linear with time and also irrelevant of context situation.  In this paper, 
we propose a new metric, named as \emph{Urgency of Information} (UoI), which is defined as the product of context-aware weight and the cost resulted from the inaccuracy of status estimation. Here, the time-dependent context-aware weight represents how crucial the status information is at the specific moment and the nonlinear cost represents the performance distortion due to the information latency.  
In other words, the UoI measures the performance degradation of a system due to the difference between the actual status and received status information. If the UoI is large, it means that the system is more \emph{urgent} for new status information.  As a special case, 
if the weight is time-invariant and context-unaware and the cost increases linearly over time, then UoI is equivalent to the conventional AoI.  Thus, UoI can be considered as an extension of AoI.  

The main contributions of this paper are summarized as follows:
\begin{enumerate}
\item Urgency of Information is proposed as a new metric for the timeliness of status updates, through which both context-based importance and non-uniform evolution of the status are evaluated.
\item The reduction of UoI subject to an average updating frequency constraint is investigated, and an adaptive updating scheme is proposed. It is theoretically proved that the adaptive scheme can give a bounded UoI. It is also shown by simulations that the proposed scheme can achieve a near-optimal UoI.
\item A transmission scheduling problem is formulated to reduce the average UoI of a multi-access network for status updates, and an adaptive scheduling policy is proposed. The decentralization of the scheduling policy is also proposed, which enables independent implementation at each terminal, and reduces potential transmission collision by fine-tuned Carrier Sensing Multiple Access with Collision Avoidance (CSMA/CA). Simulation results show that the scheduling policies have notable advantages compared to AoI-based policy in UoI reduction and control performance.
\end{enumerate}

The remainder of this paper is organized as follows. Section II introduces the concept of UoI with an example in remote tracking control. The design of adaptive updating scheme for a single terminal with average updating frequency constraint is investigated in Section III. A multi-terminal scheduling problem is formulated to reduce the UoI, and an adaptive scheduling policy along with its decentralization is proposed in Section IV. In Section V, the performance of various policies is illustrated with simulation results. Section VI concludes the paper.
                                                                                                                                                       
\section{Urgency of Information: A New Metric}
In this section, we first introduce the definition of UoI, then study the performance of a remote tracking control system in order to build the relationship between information latency and remote control, so as to understand the rationale behind UoI. 

\subsection{Definition of UoI}
Denote the actual status of a continuous signal at time $t$ by $x(t)$, the estimated status by $\hat{x}(t)$, and the estimation error at time $t$ by $Q(t) = x(t) - \hat{x}(t)$. The performance degradation caused by status estimation inaccuracy is denoted by $\delta(Q(t))$, where $\delta(\cdot)$ is a non-negative even function (e.g., norms, quadratic function). The time-varying importance of status is evaluated by context-aware weight $\omega(t)$. When the system is at a crucial situation, the corresponding context-aware weight $\omega(t)$ is large, and vice versa. 

To measure the timeliness of status updates in remote control systems, we propose a new metric called the \emph{Urgency of Information} (UoI) (previously the context-aware information lapse in \cite{previous}), which is defined as the product of weight $\omega(t)$ and cost $\delta(Q(t))$:
\begin{equation}
\label{eqn:def}
F(t) = \omega(t)\delta(Q(t)). 
\end{equation}
Due to the Markovian property of status information, the estimation error is solely dependent on the most up-to-date status update packet that has been delivered. Denote the generation time of the most up-to-date delivered status update packet before time $t$ as $g(t)$. When the packet is first generated at time $g(t)$, it is identical to the actual status. As the status evolves, the estimation error based on the packet continues to accumulate. Therefore, the estimation error at time $t$ is
\begin{eqnarray}
\label{eqn:q1}
Q(t) =  \int_{g(t)}^tA(\tau)\mathrm{d}\tau,
\end{eqnarray}
where $A(t)$ is the derivative of estimation error $Q(t)$ when there is no status information delivery, which can be any real number. There are several special cases under which UoI is equivalent to existing metrics: 
\begin{enumerate}
\item If the cost function is linear, and the estimation error $Q(t)$ increases linearly when status information is not updated (i.e., $A(t) = 1$), and the weight $\omega(t)$ is time-invariant, the UoI is equivalent to the conventional AoI. 
\item If the estimation error $Q(t)$ increases linearly when status information is not updated (i.e., $A(t) = 1$), and the weight $\omega(t)$ is time-invariant, the UoI is equivalent to the non-linear AoI defined in \cite{8000687}. 
\item If the weight $\omega(t)$ is time-invariant and $\delta(x) = x^2$, the UoI is equivalent to the squared error of status estimation. 
\end{enumerate}

Accordingly, the discrete-time expression of UoI is formulated as
\begin{equation*}F_t = \omega_t\delta\left(\sum_{\tau=g_t}^{t-1}A_{\tau}\right),\end{equation*}
where $A_t$ is the increment of error in the $t$-th time slot. The dynamic function of estimation error is
\begin{equation}
\label{eqn:dis}
Q_{t+1} = \left(1-D_t\right)Q_t + A_t + D_t\sum_{\tau=g_{t+1}}^{t-1}A_{\tau}, 
\end{equation}
where indicator $D_t=1$ represents that there is a successful status delivery in the $t$-th slot; otherwise $D_t = 0$. The last term at the right-hand side of Eq. (\ref{eqn:dis}) is resulted from the status change during the latency experienced by the latest status packet. 

The UoI implies how urgent the system is for a new status update. By reducing UoI, the context-aware timeliness of information can be better guaranteed. 

\begin{figure}
\centering\vspace*{-.15in}
\includegraphics[width=3.2in]{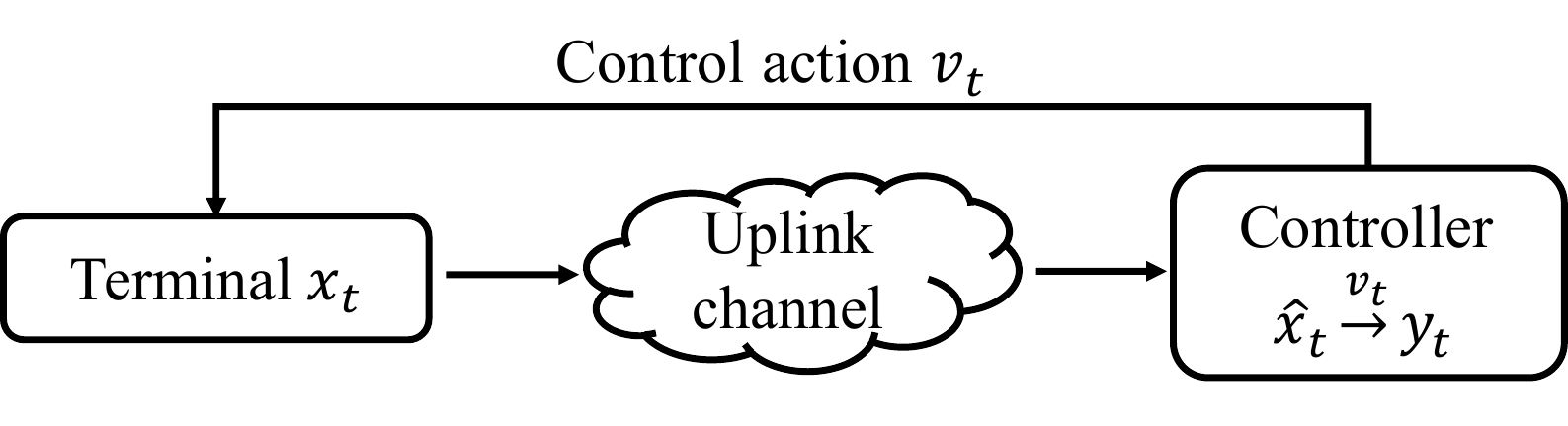}\vspace*{-.15in}
\caption{The tracking control of a linear system. }\vspace*{-.15in}
\label{fig:lqc}
\end{figure}

\subsection{Example: Tracking Control of Linear Systems}

Consider a remote control system where a controller remotely sends control decisions to a terminal based on previous status feedbacks, as shown in Fig. \ref{fig:lqc}. The status of the terminal at the $t$-th time slot is denoted as $x_t$, and the dynamic function of status evolution is:
$$x_t = ax_{t-1} + bv_t + r_t, $$
where $v_t$ is the control action at the $t$-th time slot, and noise $r_t$ is an independent and identically distributed (i.i.d.) random variable with finite variance. Without loss of generality, it is assumed that the expected value of $r_t$ is zero. To simplify the problem, it is also assumed that the control actions can be delivered to the terminal at each time slot without error, while the frequency of status feedbacks from the terminal is limited due to the constraints in the uplink channel. At each time slot, the controller decides on $v_t$ in order to keep the terminal status $\{x_t|t\in\mathcal{N}\}$ as close to the desired status $\{y_t|t\in\mathcal{N}\}$ as possible. The control performance is evaluated by the weighted squared tracking error $\omega_t(x_t-y_t)^2$, where the context-aware weight $\omega_t$ indicates the importance of tracking error at the $t$-th time slot. The objective is to minimize the weighted mean squared tracking error:
\begin{eqnarray}
\label{program:LQC}
&\min	& \limsup_{T\to\infty}\frac{1}{T}\sum_{t=0}^{T-1}\E{\omega_{t}\left(x_t - y_t\right)^2}. 
\end{eqnarray}
Due to the limitation of the uplink channel, the controller might not be able to obtain the latest status $x_{t-1}$. When this happens, it needs to first estimate the status of the terminal based on historical status information and control actions, and makes the optimal control decision $v^*_t$ based on the estimation $\hat{x}_{t-1}$. Therefore, problem (\ref{program:LQC}) consists of the optimizations of status updates and estimation-based control:
\begin{eqnarray}
\label{pro:control}
\min_{\mathrm{status~updates}}\quad\min_{v_t}\quad\limsup_{T\to\infty}\frac{1}{T}\sum_{t=0}^{T-1}\E{\omega_{t}\left(a\hat{x}_{t-1} + bv_t + r_t - y_t\right)^2}. 
\end{eqnarray}

\begin{proposition}
Problem (\ref{program:LQC}), which is to minimize the weighted squared difference between the actual state and the desired status, is equivalent to minimizing the weighted estimation error:
\begin{eqnarray}
\label{program:est}
\min_{\mathrm{status~updates}}\quad \limsup_{T\to\infty}\frac{1}{T}\sum_{t=0}^{T-1}\E{\omega_{t}\left(x_{t-1} - \hat{x}_{t-1}\right)^2}. 
\end{eqnarray}
\end{proposition}
\begin{proof}
See Appendix A. 
\end{proof}

By Proposition 1, the minimization of weighted squared tracking error in the tracking control of linear systems is equivalent to minimizing weighted estimation error, i.e. UoI. As long as the uplink channel is not perfect, the controller cannot get the instantaneous knowledge of the actual status information each time it makes control actions, which degrades remote control performance. In order to achieve a better control performance, a sophisticated status update scheme is in need to deliver status information timely based on the context information and status evolution to reduce UoI. 

\subsection{Analogue to a Queuing System}
\begin{figure}
\centering\vspace*{-.15in}
\includegraphics[width=2.1in]{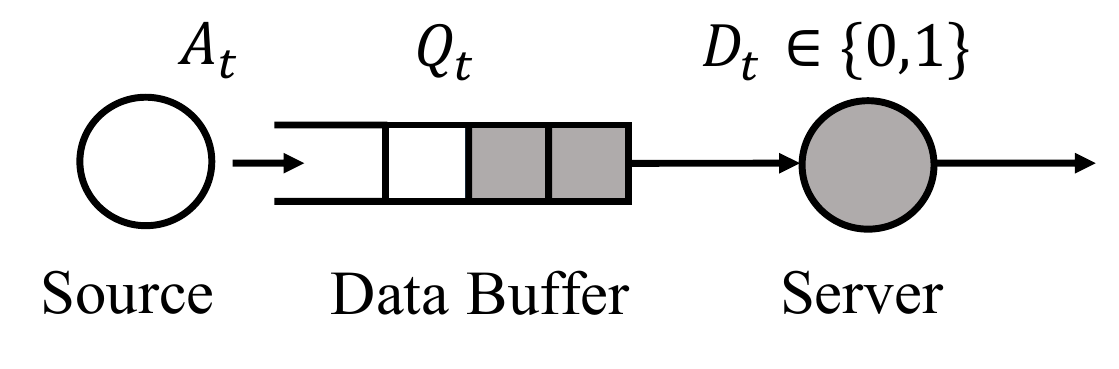}\vspace*{-.15in}
\caption{Analogue to a queuing system in which the data buffer would be emptied if there were a delivery. }\vspace*{-.15in}
\label{fig:queue}
\end{figure}

Consider a special case where the time-scale of status updates is much larger than the packet transmission time, i.e., status information can be instantaneously obtained and delivered by the terminal whenever it is scheduled. In this case, we have $g(t) = t$ when $D(t)=1$ in Eq. (\ref{eqn:q1}) and $g_t = t-1$ when $D_t=1$ in Eq. (\ref{eqn:dis}), and the two equations are thus written as
\begin{eqnarray}
\left\{
\begin{aligned}
&\frac{\mathrm{d}}{\mathrm{d}t}Q_t = A_t, &~\mathrm{if~}D_t = 0;\\
&Q_t = 0, &~\mathrm{if~}D_t = 1, 
\end{aligned}
\right.
\end{eqnarray}
and 
\begin{equation}
\label{eqn:Q}
Q_{t+1} = \left(1-D_t\right)Q_t + A_t.  
\end{equation}

The dynamic function of $Q_t$ in (\ref{eqn:Q}) is equivalent to the queuing system in Fig. \ref{fig:queue}, where the source generates $A_t$ packets at the $t$-th time slot, and the data buffer is emptied once the server completes a service, i.e., $D_t=1$. Note that different from conventional queuing systems, both ``arrivals'' $A_t$ and ``queue length'' $Q_t$ can be negative. Nonetheless, we can apply the methods in queuing theory to analyze UoI (with $\delta(x) = x^2$) as will be shown in the following sections. 

\section{Single Terminal Problem}
In this section, we look into the resource allocation for a single status update terminal, and design a context-aware update scheme to reduce UoI and improve the timeliness of status updates.

\subsection{System Model}
Consider a communication link in Fig. \ref{fig:single}, where a terminal constantly updates its status to a fusion center through a wireless channel. Due to the lack of channel resource or energy supplement, the terminal cannot send status information in every time slot. The maximum frequency of transmitting status updates to the fusion is $\rho\leq1$. The decision variable at each time slot is denoted by $U_t$. If the terminal sends its status at the $t$-th slot, we have $U_t=1$; otherwise $U_t=0$. Status packets are transmitted through a block fading channel with success probability $p$. The state of channel being good at the $t$-th slot is represented by $S_t=1$; otherwise $S_t=0$. Therefore, there is a successful status delivery at the $t$-th time slot \emph{if and only if} $D_t \triangleq U_tS_t=1$. 

\begin{figure}
\centering\vspace*{-.15in}
\includegraphics[width=3.5in]{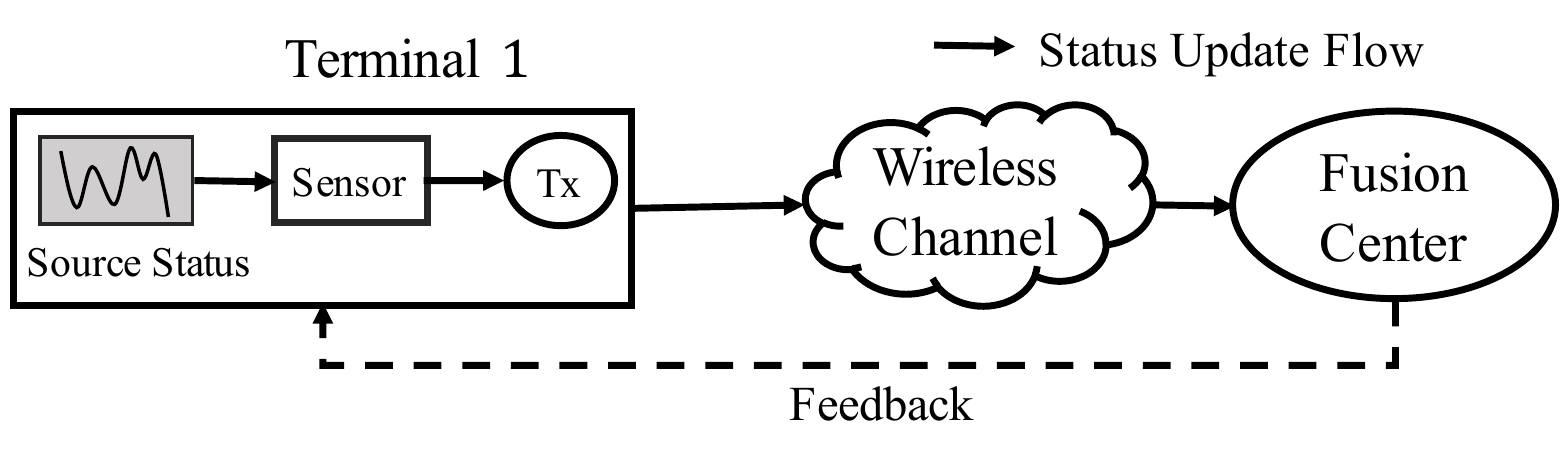}\vspace*{-.15in}
\caption{A  wireless communication system where there is a terminal updates its status to a fusion center. }\vspace*{-.15in}
\label{fig:single}
\end{figure}

Due to the randomness in status evolution, the increment $A_t$ of error is an random variable, and is assumed to have zero mean and variance $\sigma^2$. We also assumed that $A_t$ is independent of the current estimation error $Q_t$. An example is that when the monitored status follows Wiener process, the increment $A_t$ is an i.i.d Gaussian random variable. The time-varying context-aware weight $\omega_t$, which is associated with the context information, is assumed to be a random variable that has mean $\bar{\omega}$, and is independent of error $Q_t$. 

To reduce UoI so as to promote the timeliness of status updates, the terminal adaptively sends the latest status to the fusion center by deciding $U_t$ at every time slot. The average UoI minimization problem is formulated as
\begin{subequations}
\label{program:main-single}
\begin{align}
\min_{U_t} &\quad\limsup_{T\to\infty}\frac{1}{T}\E{\sum_{t=0}^{T-1}\omega_tQ_t^2}\label{obj:single}\\
\mathrm{s.t.}	&\quad \limsup_{T\to\infty}\frac{1}{T}\sum_{t=0}^{T-1}\E{U_t} \leq \rho,\label{eqn:con2}
\end{align}
\end{subequations}
where
\begin{equation}
\label{con:dynamic}
Q_{t+1} = \left(1 - U_tS_t\right)Q_t + A_t. 
\end{equation}

\subsection{Status Updating Scheme}
In order to satisfy the average transmission frequency constraint (\ref{eqn:con2}), we define the virtual queue $H_t$ as 
\begin{equation}
\label{eqn:vq}
H_{t+1} = \left[H_t - \rho + U_t\right]^+,
\end{equation}
where $[x]^+ = \max(0,x)$. The length of virtual queue can characterize the historical usage of the transmission budget: When the terminal sends a status update, the virtual queue $H_t$ increases by $(1-\rho)$; otherwise the virtual queue decreases by $\rho$. Therefore, the longer the virtual queue is, the more transmissions are performed. 
\begin{lemma}
\label{lemma:mrs}
With $H_0 < \infty$, Eq. (\ref{eqn:con2}) is satisfied as long as the virtual queue $H_t$ is mean rate stable, i.e., $\lim_{T\to\infty}\frac{\E{H_T}}{T} = 0. $
\end{lemma}
\begin{proof}
See Appendix B. 
\end{proof}

Lemma \ref{lemma:mrs} indicates that by designing a status updating scheme that keeps the virtual queue mean rate stable, the average frequency constraint (\ref{eqn:con2}) is naturally guaranteed. In order to obtain a feasible updating scheme, define a Lyapunov function as $L_t =  \frac{1}{2}VH^2_t + \theta Q_t^2$, where $\theta$ is a positive real numbers that will be further determined, and $V$ is also a positive real numbers that represents the tradeoff between convergence and performance. The Lyapunov drift is then defined as 
\begin{equation}
\label{def:drift}
\Delta_t = \E{L_{t+1} - L_t|Q_t, \omega_{t+1}, H_t}.
\end{equation}

\begin{lemma}
\label{lemma:bound-single}
Denote the penalty at the $t$-th time slot by $f_t$. If $\E{L_0}<\infty$, $\E{f_t}\geq \bar f_\mathrm{min}$ and in each time slot $t\in\{0,1,2\cdots\}$, 
\begin{eqnarray}
\label{eqn:optimization-single}
\E{L_{t+1} - L_t + f_t} \leq C,
\end{eqnarray}
then the virtual queue $H_t$ is mean rate stable, and
\begin{equation}
\label{eqn:bound-single}
\limsup_{T\to\infty}\frac{1}{T}\sum_{t=0}^{T-1}\E{f_t} \leq C. 
\end{equation}
\end{lemma}
\begin{proof}
See Appendix \ref{sec:lemma2}. 
\end{proof}

By Lemma \ref{lemma:bound-single}, as long as the sum of expected drift and penalty is no larger than a constant, the virtual queue is mean rate stable, and the average penalty is upper-bounded. Therefore, the problem of reducing average penalty while keeping the virtual queue mean rate stable can be simplified to reducing the expected sum of Lynapunov drift and penalty at each time slot. 

\begin{lemma}
\label{lemma:drift-single}
Given estimation error $Q_t$, context-aware weight at the next time slot $\omega_{t+1}$, and virtual queue length $H_t$, the sum of Lynapunov drift and expected penalty $f_t$ satisfies
\begin{eqnarray}
\label{eqn:drift-single}
\Delta_t + \E{f_t|Q_t, \omega_{t+1}, H_t}
&\leq&  \theta\sigma^2 + \frac{1}{2}V - V\rho H_t + \E{f_t|Q_t, \omega_{t+1}, H_t}\notag\\
&~& + (VH_t - \theta pQ_t^2)\E{U_t|Q_t, \omega_{t+1}, H_t}.
\end{eqnarray}
\end{lemma}
\begin{proof}
See Appendix \ref{sec:lemma3}. 
\end{proof}

Letting $f_t = \omega_{t+1}Q_{t+1}^2$, which is the UoI at the $(t+1)$-th time slot, we have
\begin{eqnarray}
\label{eqn:penalty-e}
\E{f_t|Q_t, \omega_{t+1}, H_t}&=& \omega_{t+1}\E{\left(\left(1-U_tS_t\right)Q_t + A_t\right)^2|Q_t, \omega_{t+1}, H_t}\notag\\
&=& \omega_{t+1}Q_t^2\left(1-p\E{U_t|Q_t, \omega_{t+1}, H_t}\right) + \omega_{t+1}\sigma^2.
\end{eqnarray}
Substituting Eq. (\ref{eqn:penalty-e}) into Eq. (\ref{eqn:drift-single}) yields
\begin{eqnarray}
\label{eqn:right-single}
\Delta_t + \E{f_t|Q_t, \omega_{t+1}, H_t}
&\leq& (\omega_{t+1}+\theta)\sigma^2 + \frac{1}{2}V - V\rho H_t + \omega_{t+1}Q_t^2\notag\\
&~& + (VH_t - (\omega_{t+1}+\theta) pQ_t^2)\E{U_t|Q_t, \omega_{t+1}, H_t}.
\end{eqnarray}
By minimizing the right-hand side of the above inequality, we get a status update scheme in the following form:
\begin{eqnarray}
\label{policy:original}
\begin{aligned}
\min_{U_t}&\quad (VH_t - (\omega_{t+1}+\theta) pQ_t^2)U_t \\
\mathrm{s.t.}&\quad U_t \in \{0,1\}.
\end{aligned}
\end{eqnarray}

Next, we optimize parameter $\theta$ to reduce the right-hand side of (\ref{eqn:right-single}). Note that scheme (\ref{policy:original}) always minimizes the right-hand side of (\ref{eqn:right-single}), and the stationary randomized scheme that independently updates with probability $\rho$ at each time slot is also a feasible scheme. Substituting the stationary randomized scheme into the right-hand side of (\ref{eqn:right-single}) yields
\begin{eqnarray*}
\E{L_{t+1} - L_t + f_t|Q_t, \omega_{t+1}, H_t}
&\leq& (\omega_{t+1}+\theta)\sigma^2 + \frac{1}{2}V + \left(\omega_{t+1}(1-p\rho) - \theta p\rho\right)Q_t^2.
\end{eqnarray*}
Taking expectation yields
\begin{eqnarray*}
\E{L_{t+1} - L_t + f_t|Q_t}
&\leq& (\bar\omega+\theta)\sigma^2 + \frac{1}{2}V + \left(\bar\omega(1-p\rho) - \theta p\rho\right)Q_t^2.
\end{eqnarray*}
If $\bar\omega(1-p\rho) - \theta p\rho\leq0$, the right-hand side is no greater than a constant. Therefore, letting $\theta = \frac{\bar\omega(1-p\rho)}{p\rho}$ yields
\begin{eqnarray}
\label{eqn:bound-final-single}
\E{L_{t+1} - L_t + f_t|Q_t}
\leq \frac{\bar\omega\sigma^2}{p\rho} + \frac{1}{2}V.
\end{eqnarray}
Then scheme (\ref{policy:original}) becomes solving
\begin{subequations}
\label{p:main-single}
\begin{align}
\min_{U_t}&\quad\left(VH_t - \left(\omega_{t+1}-\bar\omega+\frac{\bar\omega}{p\rho}\right)pQ_t^2\right)U_t\label{eqn:policy-obj}\\
\mathrm{s.t.}&\quad U_t \in \{0,1\}.
\end{align}
\end{subequations}
Next we provide the details of the solution to (\ref{p:main-single}).
\begin{definition}
In a single-terminal status update system with an average updating frequency constraint, the update index $J_t$ at the $t$-th time slot is defined as
\begin{equation}
\label{eqn:ui}
J_t=\left(\omega_{t+1}-\bar\omega+\frac{\bar\omega}{p\rho}\right)pQ_t^2.
\end{equation}
\end{definition}
\begin{proposition}
The solution to scheme (\ref{p:main-single}) is
\begin{eqnarray}
\label{eqn:u-single}
U_t=\left\{
\begin{aligned}
&1, &&\mathrm{~if~}J_t > VH_t,\\
&0, &&\mathrm{~if~}J_t \leq VH_t.
\end{aligned}
\right.
\end{eqnarray}
\end{proposition}

The update index equals to the reduction of the sum of expected future UoI with a status transmission. According to Eq. (\ref{eqn:u-single}), the status update terminal first computes its update index $J_t$ by estimation error $Q_t$ and next-step context-aware weight $\omega_{t+1}$ at each time slot. If the update index is larger than $VH_t$, there will be a status transmission. In other words, update index measures the necessity of status transmission with the consideration of the current context and status estimation error, while the virtual queue is a dynamic threshold that ensures to meet the average status update frequency constraint.

\subsection{Performance Analysis}
\begin{theorem}
\label{thm:single}
Under updating scheme (\ref{eqn:u-single}), the average status update frequency constraint (\ref{eqn:con2}) is satisfied, and the average UoI is upper bounded as
\begin{eqnarray}
\limsup_{T\to\infty}\frac{1}{T}\sum_{t=0}^{T-1}\omega_tQ_t^2\leq\frac{\bar{\omega}\sigma^2}{p\rho} + \frac{1}{2}V. 
\end{eqnarray}
\end{theorem}
\begin{proof}
Taking expectation over both sides of inequality (\ref{eqn:bound-final-single}) yields
$$\E{L_{t+1} - L_t + f_t}
\leq \frac{\bar\omega\sigma^2}{p\rho} + \frac{1}{2}V.$$
By the definition of penalty $f_t$ and Lemma \ref{lemma:bound-single}, the theorem is hereby proved. 
\end{proof}

According to Theorem \ref{thm:single}, under updating scheme (\ref{eqn:u-single}), a smaller parameter $V$ leads to a lower the UoI upper-bound. However, when parameter $V$ is small, the updating scheme (\ref{eqn:u-single}) pays little attention to the updating frequency constraint, which could lead to a severe overuse of transmission budget at some critical period, and result in service outage afterwards. 

\section{Multi-Terminal Scheduling}
In this section, we investigate the reduction of average UoI in an uplink status update system with multiple terminals sharing wireless resources, so that the interaction between different terminals needs to be considered. By scheduling the status updates of multiple terminals, the average UoI of the system is optimized. 

\subsection{System Model}
\begin{figure}
\centering\vspace*{-.15in}
\includegraphics[width=2.8in]{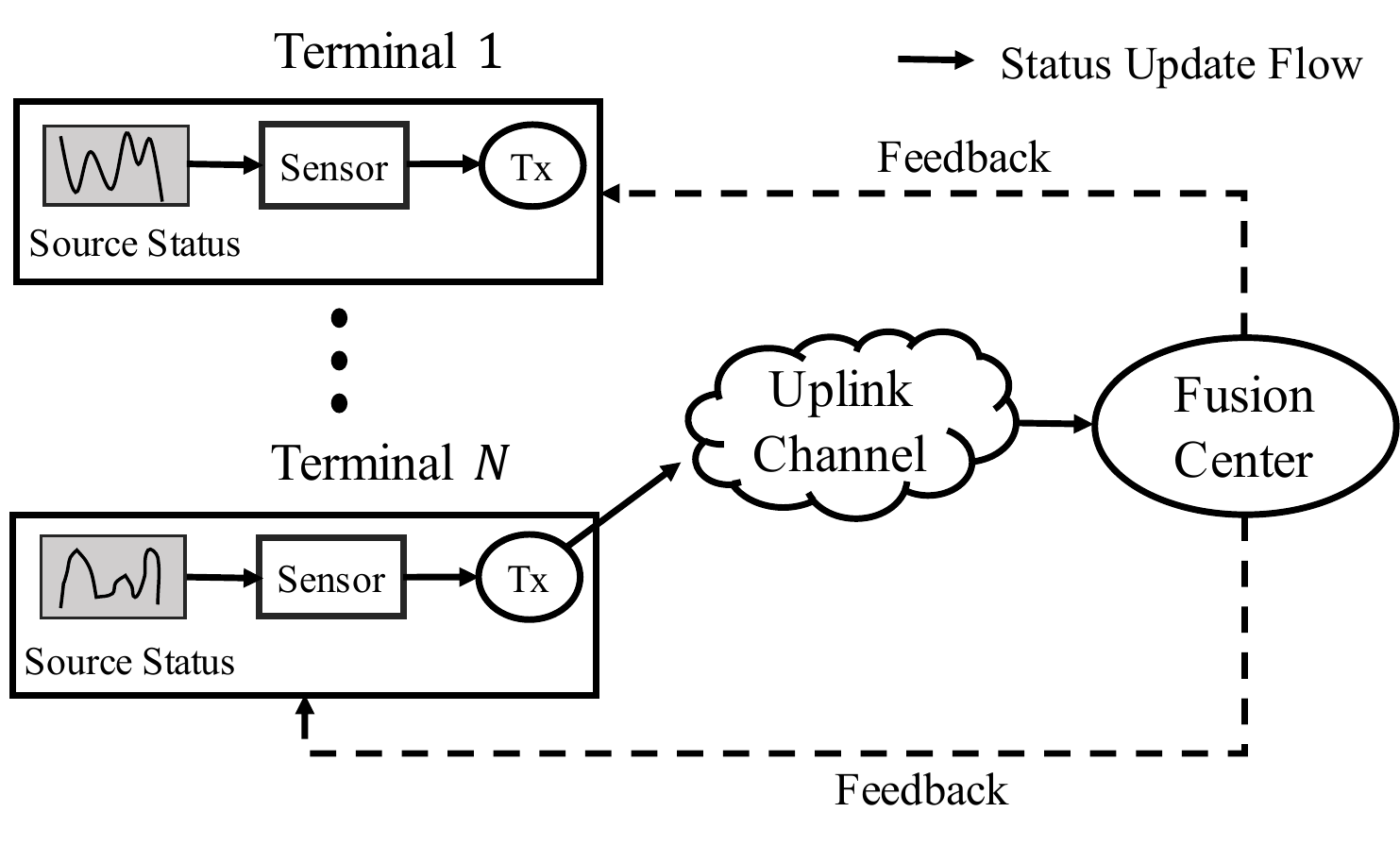}\vspace*{-.15in}
\caption{A wireless communication system consisting of $N$ terminals for time-sensitive status updates. }\vspace*{-.15in}
\label{fig:control}
\end{figure}

Consider a wireless communications system with $N$ status updates terminals, as illustrated in Fig. \ref{fig:control}. Due to limited channel resources, at each time slot at most $K$ terminals can be scheduled to transmit their status update packets simultaneously. The rest assumptions about status dynamics, context-aware weight, scheduling decisions and channel states are the same as those in Sec. III, with an additional subscripted denoting the index of the underlying terminal. For example, scheduling decisions at the $t$-th time slot is denoted by a vector $\boldsymbol{U}_t = \left(U_{1,t}, U_{2,t}, \cdots, U_{N,t}\right)$, where $U_{i,t}=1$ means that the $i$-th status updates terminal is scheduled at the $t$-th slot. The transmission success probability of the $i$-th terminal is $p_i$, and the state of the $i$-th status updates terminal's channel being good at the $t$-th slot is represented by $S_{i,t}=1$. Estimation error increment $A_{i,t}$ has zero mean and variance $\sigma_i^2$ and is independent among different time slots and terminals. The context-aware weight $\omega_{i,t}$ has mean $\bar{\omega}_i$, and is independent of error $Q_{i,t}$. 

To improve the timeliness of status updates, the scheduler tries to adaptively schedule the transmission of each terminal by reducing the average UoI. The average UoI minimization problem is formulated as
\begin{subequations}
\label{program:main}
\begin{eqnarray}
\label{program:main1}
\min_{\boldsymbol{U}_t} &&\limsup_{T\to\infty}\frac{1}{NT}\E{\sum_{t=0}^{T-1}\sum_{i=1}^{N} \omega_{i,t}Q_{i,t}^2}\\
\mathrm{s.t.} &&\sum_{i=1}^N U_{i,t} \leq K, \forall t \in \mathbb{N}, \label{eqn:con}
\end{eqnarray}
\end{subequations}
where the dynamic function of estimation error $Q_{i,t}$ is
\begin{equation}
\label{eqn:dynamic}
Q_{i,t+1} = \left(1-U_{i,t}S_{i,t}\right)Q_{i,t} + A_{i,t}.
\end{equation}
Eq. (\ref{eqn:con}) poses the constraint on maximum number of terminals to update their status. By carefully designing scheduling decision $\boldsymbol{U}_t$ at each time slot, we try to reduce the average UoI. 

\subsection{Scheduling Scheme}
Define a Lyapunov function as $L_t =  \sum_{i=1}^N\theta_i Q_{i,t}^2$, where $\theta_i$ are positive. At the meantime, the Lyapunov drift is defined as 
\begin{equation}
\label{def:drift-multi}
\Delta_t = \E{L_{t+1} - L_t|\boldsymbol{Q}_t, \boldsymbol{\omega}_{t+1}}.
\end{equation}
Let the penalty at the $t$-th time slot be the UoI at the $(t+1)$-th time slot, i.e., $f_t = \sum_{i=1}^N\omega_{i,t+1}Q_{i,t+1}^2$. Thus, we have the following lemma:
\begin{lemma}
Given $\boldsymbol{Q}_t$ and $\boldsymbol{\omega}_{t+1}$ at the $t$-th time slot, the drift-plus-penalty is 
\begin{eqnarray}
\label{eqn:driftpluspenalty-multi}
\Delta_t + \E{f_t|\boldsymbol{Q}_t, \boldsymbol{\omega}_{t+1}}
&=&  - \sum_{i=1}^N\left(\theta_i+\omega_{i,t+1}\right)p_iQ_{i,t}^2\E{U_{i,t}|\boldsymbol{Q}_t, \boldsymbol{\omega}_{t+1}}\notag\\
&~& + \sum_{i=1}^N\left(\theta_i+\omega_{i,t+1}\right)\sigma_i^2 + \sum_{i=1}^N\omega_{i,t+1}Q_{i,t}^2.
\end{eqnarray}
\end{lemma}
\begin{proof}
See Appendix \ref{sec:lemma4}. 
\end{proof}

By minimizing the right-hand side of Eq. (\ref{eqn:driftpluspenalty-multi}), we obtain a scheduling scheme: 
\begin{eqnarray}
\label{policy:0-multi}
\begin{aligned}
&\max_{\boldsymbol{U}_{t}}~ && \sum_{i=1}^N\left(\theta_i+\omega_{i,t+1}\right)p_iQ_{i,t}^2U_{i,t}\\
&\mathrm{s.t.}&&\sum_{i=1}^N U_{i,t} \leq K.
\end{aligned}
\end{eqnarray}

\subsection{UoI Upper Bound}
\begin{theorem}
\label{thm:bound-multi}
Let $\theta_i = \frac{\bar\omega_i(1-p_i\pi_i)}{p_i\pi_i}$, where $\boldsymbol{\pi} = \left(\pi_1,\pi_2,\cdots,\pi_N\right)$ is a feasible stationary randomized scheme that scheduling the $i$-th terminal with probability $\pi_i$ at each time slot. The average UoI under scheme (\ref{policy:0-multi}) is upper-bounded by
\begin{eqnarray}
\label{eqn:bound-multi}
\limsup_{T\to\infty}\frac{1}{NT}\sum_{t=0}^{T-1}\sum_{i=1}^N\omega_{i,t}Q_{i,t}^2\leq\frac{1}{N}\sum_{i=1}^N\frac{\bar\omega_i\sigma_i^2}{p_i\pi_i}. 
\end{eqnarray}
\end{theorem}
\begin{proof}
See Appendix \ref{sec:the2}. 
\end{proof}

When parameter $\theta_i = \frac{\bar\omega_i(1-p_i\pi_i)}{p_i\pi_i}$, scheme (\ref{policy:0-multi}) is equivalent to
\begin{eqnarray}
\label{policy:main-multi-2}
\begin{aligned}
\max_{\boldsymbol{U}_t} &\quad \sum_{i=1}^N\left(\frac{1}{p_i\pi_i} - 1 + \frac{\omega_{i,t+1}}{\bar\omega_i}\right)p_i\bar\omega_iQ_{i,t}^2U_{i,t}\\
\mathrm{s.t.}&\quad\sum_{i=1}^N U_{i,t} \leq K. 
\end{aligned}
\end{eqnarray}
Theorem \ref{thm:bound-multi} gives the UoI upper-bound under scheme (\ref{policy:main-multi-2}). In addition, according to Eq. (\ref{eqn:drift-bound-multi}), parameter $\theta_i = \frac{\bar\omega_i(1-p_i\pi_i)}{p_i\pi_i}$ leads to a scheduling scheme with the least upper-bound given $\boldsymbol{\pi}$. 

\begin{figure}[htbp]
\centering\vspace*{-.15in}
\includegraphics[width=2.5in]{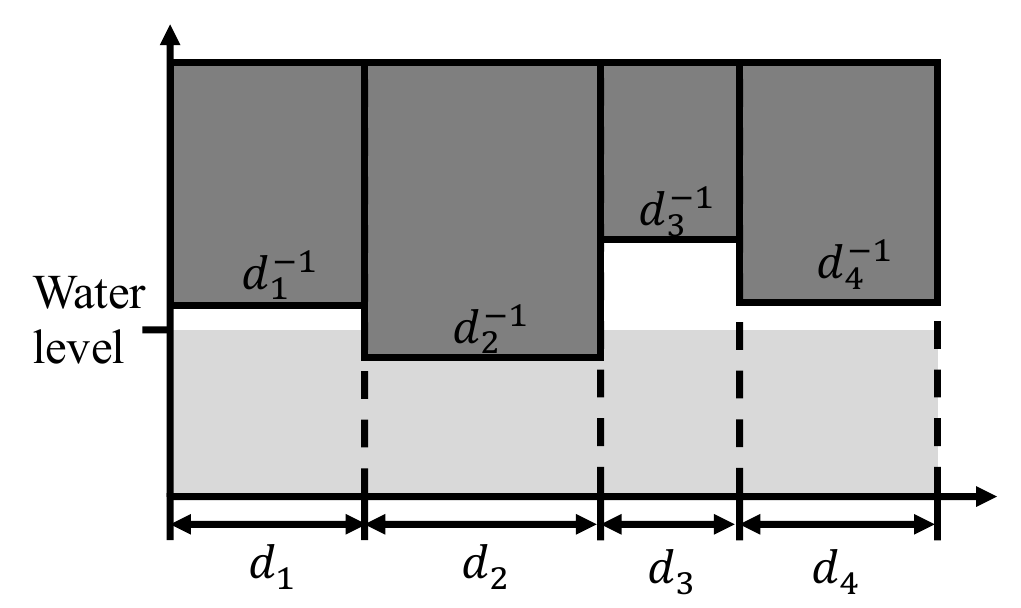}\vspace*{-.15in}
\caption{Solving problem (\ref{pro:pi}) with a water-filling algorithm. }\vspace*{-.15in}
\label{fig:wf}
\end{figure}

\subsection{Minimization of the Upper-Bound}
\label{sec:pi}
To minimize the UoI upper-bound at the right-hand side of Eq. (\ref{eqn:bound-multi}), the randomized scheme $\boldsymbol{\pi}$ needs to be optimized. The problem is described as
\begin{subequations}
\label{pro:pi}
\begin{align}
\min_{\boldsymbol{\pi}} &\quad \sum_{i=1}^N\frac{\bar\omega_i\sigma_i^2}{p_i\pi_i}\label{eqn:obj-pi}\\
\mathrm{s.t.}&\quad\sum_{i=1}^N \pi_i \leq K, \label{con:piK}\\
&\quad\pi_i \in [0,1], \forall i \in \{1,2,\cdots,N\}. \label{con:pi}
\end{align}
\end{subequations}
Since the objective function (\ref{eqn:obj-pi}) in program (\ref{pro:pi}) is convex, and the feasible region defined by Eqs. (\ref{con:piK})--(\ref{con:pi}) is a convex set, program (\ref{pro:pi}) is a convex problem. By convex optimization \cite{convex}, the problem can be numerically solved. In addition, letting $d_i=\sqrt{\frac{\bar\omega_i\sigma_i^2}{p_i}}$, there are three more straightforward methods to obtain the solution to problem (\ref{pro:pi}): 
\begin{enumerate}
\item If $\frac{\max_id_i}{\sum_{j=1}^N d_j}\leq\frac{1}{K}$, by Cauchy-Schwarz Inequality, the solution is 
\begin{eqnarray*}
\pi_i = \frac{d_i}{\sum_{j=1}^N d_j}K, \forall i \in \{1,2,\cdots,N\}. 
\end{eqnarray*}
\item In a more general case, the Karush-Kuhn-Tucher (KKT) conditions are
\begin{eqnarray*}
\left\{
\begin{aligned}
&\lambda_0\left(\sum_{i=1}^N \pi_i - K\right) = 0,\\
&\lambda_i\left(\pi_i-1\right) = 0, \forall i \in \{1,2,\cdots,N\}\\
&\lambda_0 + \lambda_i - \frac{d_i^2}{\pi_i^2} = 0, \forall i \in \{1,2,\cdots,N\}. 
\end{aligned}
\right.
\end{eqnarray*}
\item The solution to program (\ref{pro:pi}) is equivalent to a water-filling problem in Fig. \ref{fig:wf}, where the total volume of the water is $K$, the width of each bar is $d_i=\sqrt{\frac{\bar{\omega}_i\sigma_i^2}{p_i}}$, and the height of each bar is $d_i^{-1}$, such that the maximum volume of water in each bar is $1$. The volume of water in each bar gives the probability $\pi_i$ in the randomized stationary scheme $\boldsymbol{\pi}$.  
\end{enumerate}

\subsection{Solution to (\ref{policy:main-multi-2})}


Since the objective function of scheme (\ref{policy:main-multi-2}) is a linear combination of scheduling decisions $U_{i,t}$, the solution has a rather simple structure described as follows.  
\begin{definition}
In a multi-terminal status update system, the update index $J_{i,t}$ of the $i$-th terminal at the $t$-th time slot is defined as
$$J_{i,t} = \left(\bar{\omega}_i\left(\frac{1}{p_i\pi_i}-1\right)+\omega_{i,t+1}\right)p_iQ_{i,t}^2.$$
\end{definition}
\begin{proposition}
The solution to scheme (\ref{policy:main-multi-2}) is to schedule the $K$ terminals with the largest update indices $J_{i,t}$. 
\end{proposition}

By Proposition 3, a higher context-aware weight $\omega_{i,t+1}$ or cost $Q_{i,t}^2$ brings a larger chance of transmission, since its update index is larger. 

\subsection{CSMA-Based Distributed Scheduling Scheme}

\begin{figure}[htbp]
\centering\vspace*{-.15in}
\includegraphics[width=3.7in]{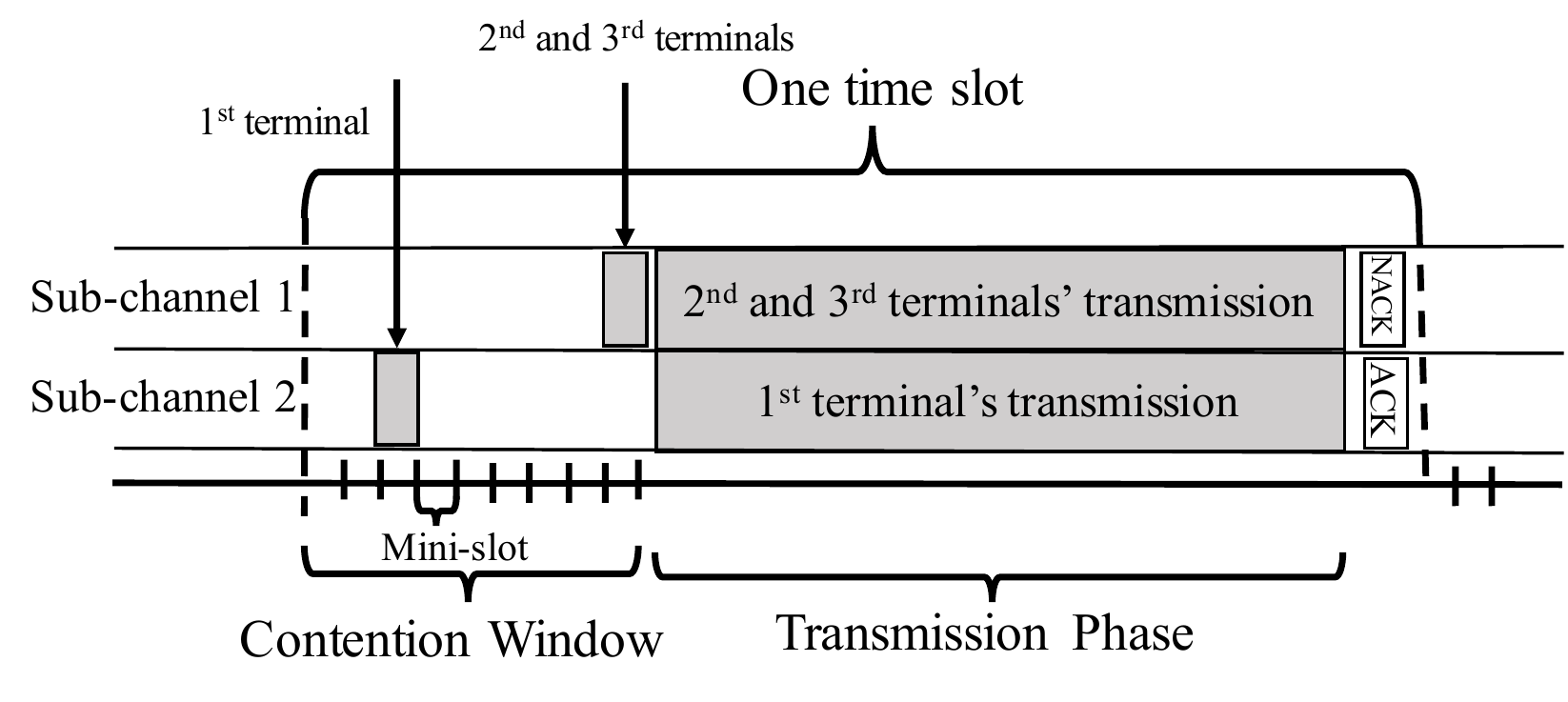}\vspace*{-.15in}
\caption{Distributed implementation of context-aware adaptive scheduling scheme for status updates. }\vspace*{-.15in}
\label{fig:distributed}
\end{figure}

Although scheme (\ref{policy:main-multi-2}) is simple-structured, it requires context and status information of all terminals to make a centralized scheduling decision, which is impractical for uplink status updates in multi-access networks. Therefore, a distributed scheduling scheme that does not require global real-time information is essential.

\begin{algorithm}[htbp]
\caption{Distributed Scheduling. }
\label{alg:distributed}
\begin{algorithmic}[1]
\Require $p_i$, $\pi_i$, $\bar\omega_i$, $\Delta J$;
\For{each time slot $t$}
\State calculate $J_{i,t} \leftarrow \left(\bar{\omega}_i\left(\frac{1}{p_i\pi_i}-1\right)+\omega_{i,t+1}\right)p_iQ_{i,t}^2$;
\If{$J_{i,t} > J_{\mathrm{thrshld},t}$}
\State generate uniformly distributed random backoff time $l_i\in[0, W-1]$;
\State $\tau\leftarrow 1$, $j\leftarrow 1$;
\While{$\tau \leq l_i$ and $j \leq K$}
\State listen to the $j$-th sub-channel;
\If{there is an intention message sent over the $j$-th sub-channel}
\State $j\leftarrow j+1$;
\EndIf
\State $\tau \leftarrow \tau+1$;
\EndWhile
\If{there are idle sub-channels at the $(l_i+1)$-th mini-slot, i.e., $\tau=l_i+1$ and $j\leq K$}
\State sends an intention message at the $(l_i+1)$-th mini-slot in the $j$-th sub-channel;
\State $j\leftarrow j+1$;
\While{$\tau\leq W$ and $j\leq K$}
\State listen to the $j$-th sub-channel;
\If{there is an intention message sent over the $j$-th sub-channel}
\State $j\leftarrow j+1$;
\EndIf
\State $\tau \leftarrow \tau+1$;
\EndWhile
\State the transmission phase begins, and sends status information with the corresponding sub-channel;
\State receive ACK/NACK from the fusion center;
\EndIf
\EndIf
\If{there are idle sub-channels in the time slot}
\State $J_{\mathrm{thrshld},t+1} = J_{\mathrm{thrshld},t} - \Delta J$;
\ElsIf{the contention window length is smaller than $\frac{K}{K+1}(W+1)$ mini-slots}
\State $J_{\mathrm{thrshld},t+1} = J_{\mathrm{thrshld},t} + \Delta J$;
\EndIf
\EndFor
\end{algorithmic}
\end{algorithm}

In this section, we propose to adaptively schedule the status updates in a decentralized manner with a CSMA/CA-based access scheme. It is observed that the goal of scheme (\ref{policy:main-multi-2}) is to schedule the $K$ terminals with the largest value of update indicator $J_{i,t}$, while the computation of $J_{i,t}$ does not require real-time information from other terminals, i.e., it is naturally decoupled. Hence, we propose a threshold-based scheme (Algorithm \ref{alg:distributed}), in which each terminal compares its update indices $J_{i,t}$ to a dynamic threshold $J_{\mathrm{th}, t}$. A terminal is turn to \emph{active} mode and competes for channel access if $J_{i,t}>J_{\mathrm{th},t}$. To resolve collision, each time slot consists of a transmission phase and a contention window that has a maximum length of $W$ mini-slots. The $i$-th terminal generates a random variable $l_i$ that is uniformly distributed in the range of $[0, W-1]$ if it is active, and senses the sub-channels one-by-one for the first $l_i$ mini-slot in the contention window. If there are idle sub-channels after the $l_i$ mini-slots expires, the terminal selects the next idle sub-channel and sends an message to reserve the sub-channel at the $(l_i+1)$-th mini-slot for the transmission phase. If all the sub-channels are occupied or the window size limit is reached, the contention window is closed. After the contention window, status information is transmitted at the terminals' reserved sub-channels. At the end of each time slot, the fusion center sends acknowledge (ACK/NACK) back to each terminal. An illustration of the distributed scheme in a system with two sub-channels and three active terminals are shown in Fig. \ref{fig:distributed}. The 1st terminal's back-off time is $l_1=2$ mini-slots, while the backoff time of both the 2nd and the 3rd terminal is 8 mini-slots. At the 3rd mini-slots, the 1st terminal sends a intention message in the 1st sub-channel, so that the 2nd and the 3rd terminals know that the 1st sub-channel is occupied and start to listen to the 2nd sub-channel. However, since the 2nd and the 3rd terminals have exactly the same back-off time, their transmissions collide in the 2nd sub-channel, and will receive NACK from the fusion center.

In \cite{DBLP:journals/corr/abs-1803-08189}, a threshold-based scheme for distributed AoI minimization is proposed. The threshold in \cite{DBLP:journals/corr/abs-1803-08189} is static and is optimized beforehand by bi-section search, which also works for the distributed scheduling scheme in this section. However, a fixed threshold does not only need extra computation for optimization via simulation, but also lacks the ability to adapt to varying context $\omega_{i,t}$. Instead, we propose a dynamic threshold that is updated locally at the end of each time slot. Ideally, only the $K$ terminals with the largest update indices are active at each time slot. Therefore, the threshold can be dynamically adjusted according to the proposed Algorithm 1 based on the following intuitions:
\begin{enumerate}
\item If there are idle sub-channels in the transmission phase, it is more likely that the number of active terminals is too small, so the threshold should be decreased (see steps 21-22);
\item If all the sub-channels are occupied and the contention window is short, it is more likely that there are too many active terminals, so the threshold should be increased (see steps 23-24). Here the contention window is deemed too short if it is smaller than the expected contention window length, which is given by Lemma 5. 
\end{enumerate}
\begin{lemma}
If there are $K$ active terminals and no idle sub-channel, the expected length of contention window is $\E{W_K}=\frac{K}{K+1}(W+1)$.
\end{lemma}
\begin{proof}
See Appendix D.
\end{proof}

\begin{figure}[htbp]
\centering\vspace*{-.15in}
\includegraphics[width=4.7in]{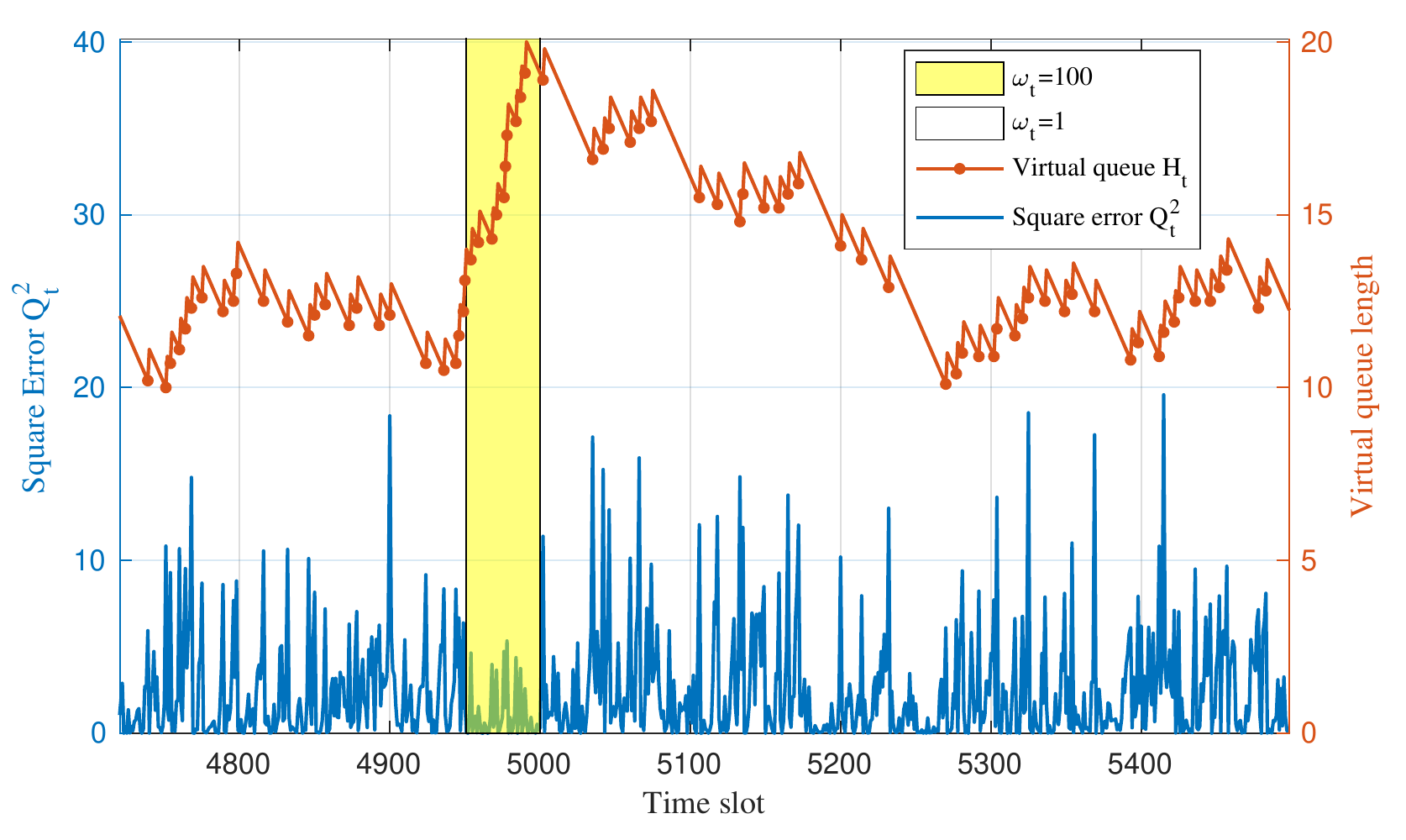}\vspace*{-.15in}
\caption{Virtual queue length and squared estimation error in a single terminal updating problem. }\vspace*{-.15in}
\label{fig:ss}
\end{figure}

\section{Numerical Results}
In the simulations, the fusion center receives updates from terminals whose source status evolves as a standard Wiener process, and the transmission of a status update packet takes $1$ms. The increment of estimation error is normalized, so that the increment in $1$ms is an i.i.d. standard Gaussian random variable. Apart from the distributed scheduling scheme, overheads are neglected  such that the duration of each time slot is $1$ms. We evaluate the performance of proposed scheduling schemes under several typical scenarios. 

\subsection{Single Terminal}
\subsubsection{Adaptive Status Updates}
Fig. \ref{fig:ss} plots a sample evolution of the virtual queue length and squared estimation error. In the single terminal system, the terminal is allowed to transmit in one fourth of the time slots, i.e., $\rho = 0.25$, with a transmission success probability of 1 for simplicity. The context-aware weight is set to $1$ in the former $4950$ time slots (with white background in the figure) of every $5000$ time slots, and set to $100$ in the latter $50$ time slots (highlighted in yellow). Parameter $V$ is set to 1. The difference in context-aware weight indicates that the latter $50$ time slots are critical in which the timeliness of status information is faced with a highly strict standard, while the others are in an ordinary period. The solid red dots indicate a status delivery. As shown in the figure, the virtual queue length significantly increases over the critical period, meaning that the terminal is updating more frequently. Accordingly, the squared estimation error is much lower over the critical period. After the critical period, the length of virtual queue drops because the terminal transmits less frequently.

\subsubsection{Tradeoff Between Updating Frequency and UoI}
In Fig. \ref{fig:tradeoff}, the average UoI under different status updating frequency constraints and parameter $V$ is plotted. Transmission success probability is set to $0.8$. The context-aware weight at each time slot is i.i.d. with probability $0.01$ being $100$ and probability $0.99$ being $1$. As the figure shows, since a larger $V$ gives a higher priority on guaranteeing the average updating frequency constraint, it has worse status update timeliness. Especially when $V$ is larger than $512$ and $\rho$ is larger than $0.5$, the curve is hardly smooth because the actual updating frequency is much lower than the frequency bound $\rho$. In addition, as the available channel resources for status updates, i.e., $\rho$, increases, the average UoI is reduced, which implies the tradeoff between resource usage and status update timeliness.

\begin{figure}[htbp]
\centering\vspace*{-.15in}
\includegraphics[width=3in]{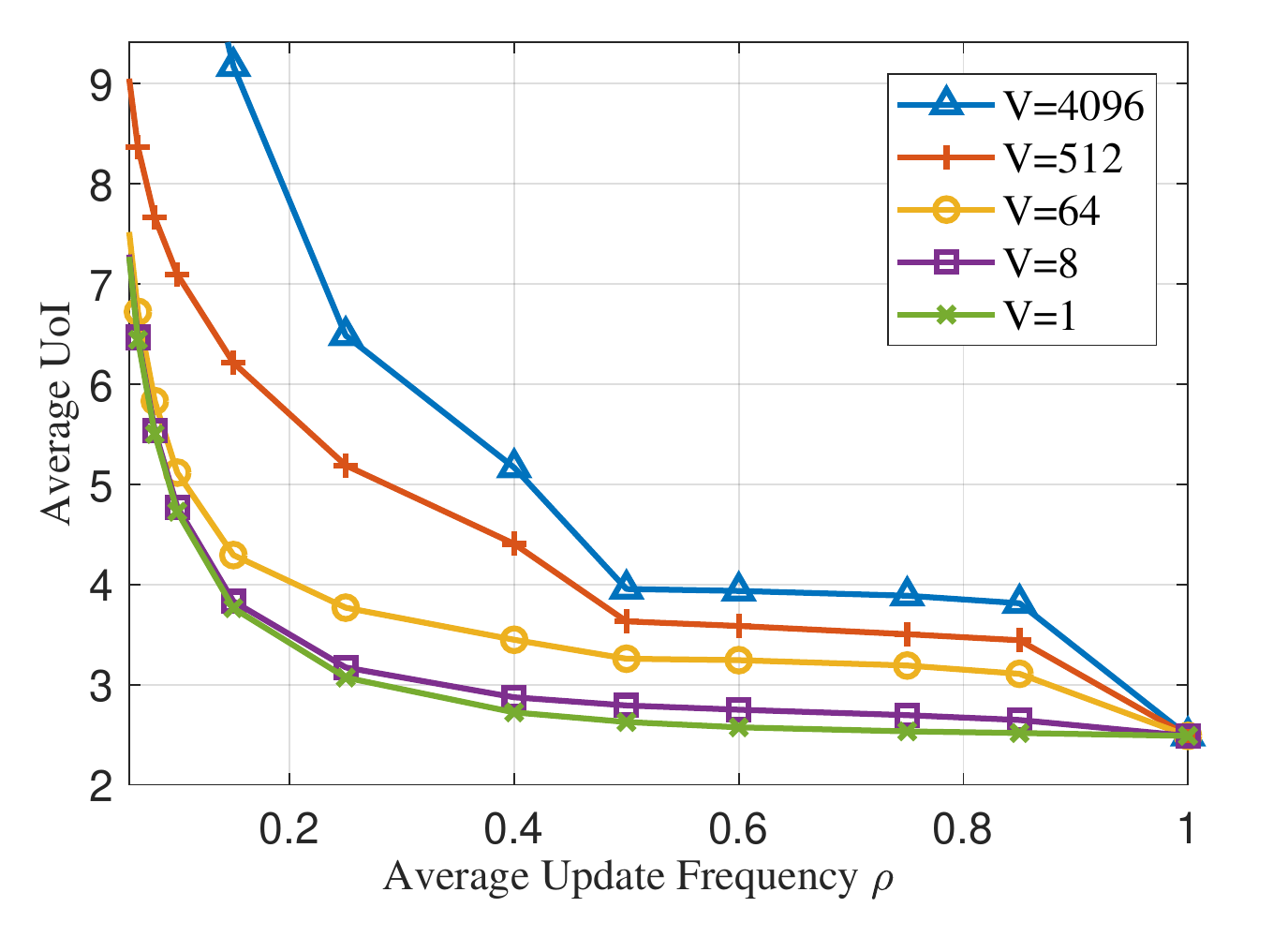}\vspace*{-.15in}
\caption{Tradeoff between resource usage and status updates performance. }\vspace*{-.15in}
\label{fig:tradeoff}
\end{figure}

\subsubsection{Performance Comparison}
Since the increment of estimation error is i.i.d. Gaussian random variable, and the context-aware weight $\omega_t$ are i.i.d., the system is Markovian with state at each time slot being $(Q_t, \omega_t, \omega_{t+1})$. Therefore, by relative value iteration \cite{bertsecas}, the Pareto optimal tradeoff between the average update frequency $\rho$ and average UoI can be obtained. Similarly, the AoI-optimal update scheme under each update frequency constraint is also obtained. Fig. \ref{fig:single-perf} illustrates the average UoI under the AoI-optimal scheme, the UoI-optimal scheme, and the proposed adaptive scheme. As the figure shows, although the adaptive scheme has a simple structure and low computation complexity compared to the optimal scheme, it can achieve a near-optimal UoI, which substantially outperforms the AoI-optimal scheme. 

\begin{figure}[htbp]
\centering\vspace*{-.15in}
\includegraphics[width=3in]{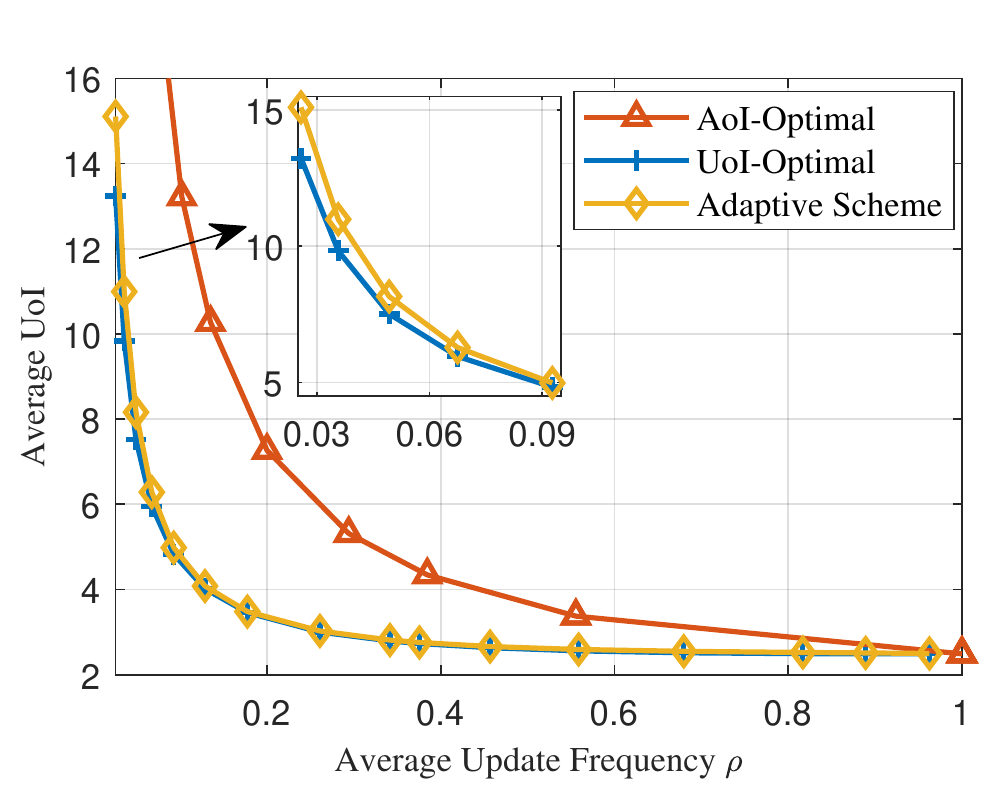}\vspace*{-.15in}
\caption{UoI under the proposed adaptive scheme, UoI-optimal scheme, and AoI-optimal scheme. }\vspace*{-.15in}
\label{fig:single-perf}
\end{figure}

\subsection{Multi-Terminal Scheduling}
In the system with multiple status update terminals, we compare the performance of several scheduling schemes in terms of average UoI, error-bound violation probability, and their performance in a remote control system, specifically \emph{CartPole}\cite{CartPole}, in order to evaluate the effectiveness of the proposed metric and the corresponding scheduling scheme. There are four scheduling schemes compared in the simulation, i.e., 
\begin{itemize}
\item Round robin: Terminals are scheduled to update one by one in a fixed order. 
\item AoI based: Schedule the top $K$ terminals with the largest value of $p_i\Delta_{i,t}(\Delta_{i,t}+1)$, where $\Delta_{i,t}$ denotes the AoI of the $i$-th terminal at the $t$-th time slot. Ref. \cite{Jiang} shows that this scheme is asymptotically AoI-optimal when channel failure probability goes to zero and the number of terminals is large. 
\item Centralized scheduling scheme: As described in scheme (\ref{policy:main-multi-2}), it requires global information at scheduling. 
\item Distributed scheduling scheme: As described in Algorithm \ref{alg:distributed}, it is implemented locally at each terminal. Note that as the window size limit $W$ increases, the probability of contention drops, while the length of a time slot (i.e., the time interval of scheduling) becomes longer. We will compare the performance of the distributed scheme under different window sizes in the first set of simulations. The threshold increment is set to $\Delta J = \frac{1}{N}\sum_{i=1}^{N}\bar\omega_i\sigma_i^2$, which is the expected growth of average UoI if none of the terminals successfully deliver its status. 
\end{itemize}

In the simulations, there are $K=2$ sub-channels in the status update system. The transmission success probability of the $i$-th terminal is $p_i=0.7 + \frac{0.3(i-1)}{N-1}$, so that the average success probability remains the same as the amount of terminals varies. The context aware weight is i.i.d. with probability 0.05 being 100 and probability 0.95 being 1. 

\subsubsection{Distributed Scheduling Scheme}
\begin{figure}[htbp]
\centering\vspace*{-.15in}
\includegraphics[width=3.2in]{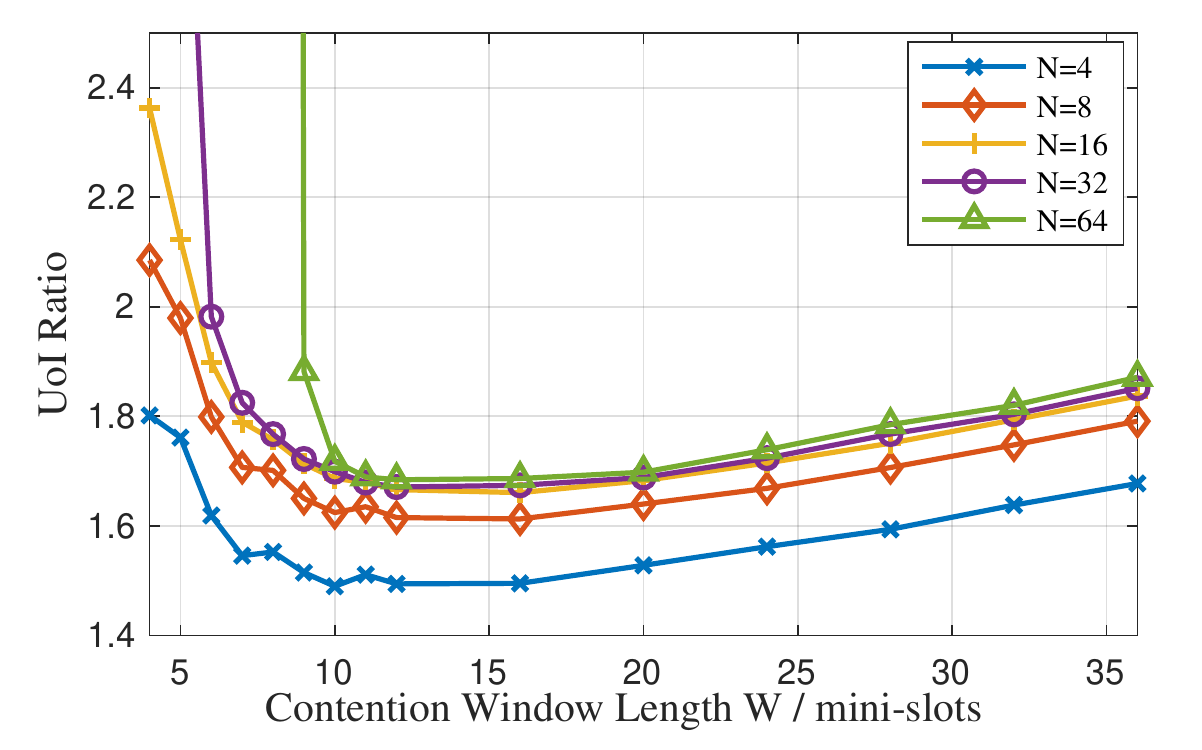}\vspace*{-.15in}
\caption{The ratio of UoI under distributed scheduling to UoI under centralized scheduling. }\vspace*{-.15in}
\label{fig:distributed-result}
\end{figure}

Under the distributed scheduling scheme, a time slot is slightly longer due to the existence of contention window, which consists of $W$ mini-slots. In the simulation, the length of a mini-slot is set to $10\mu$s (For comparison, in IEEE 802.11g, the length of a mini-slot is $9\mu$s), so that the length of a time slot in the distributed scheme is $\left(1+\frac{W}{100}\right)$ms. Therefore, the variance of estimation error's increment is $\left(1+\frac{W}{100}\right)$. In Fig. \ref{fig:distributed-result}, the ratio of UoI under decentralized scheduling to its centralized counterpart is illustrated. When the contention window size is small, the collision probability is high, so the distributed scheme performs much worse. As the window size increases, the length of a time slot grows as well, which reduces the frequency of scheduling. Overall, the window size $W=16$ gives a relatively good performance for a system with $K=2$ sub-channels.

\begin{figure}[htbp]
\centering
\vspace*{-.15in}
\includegraphics[width=3in]{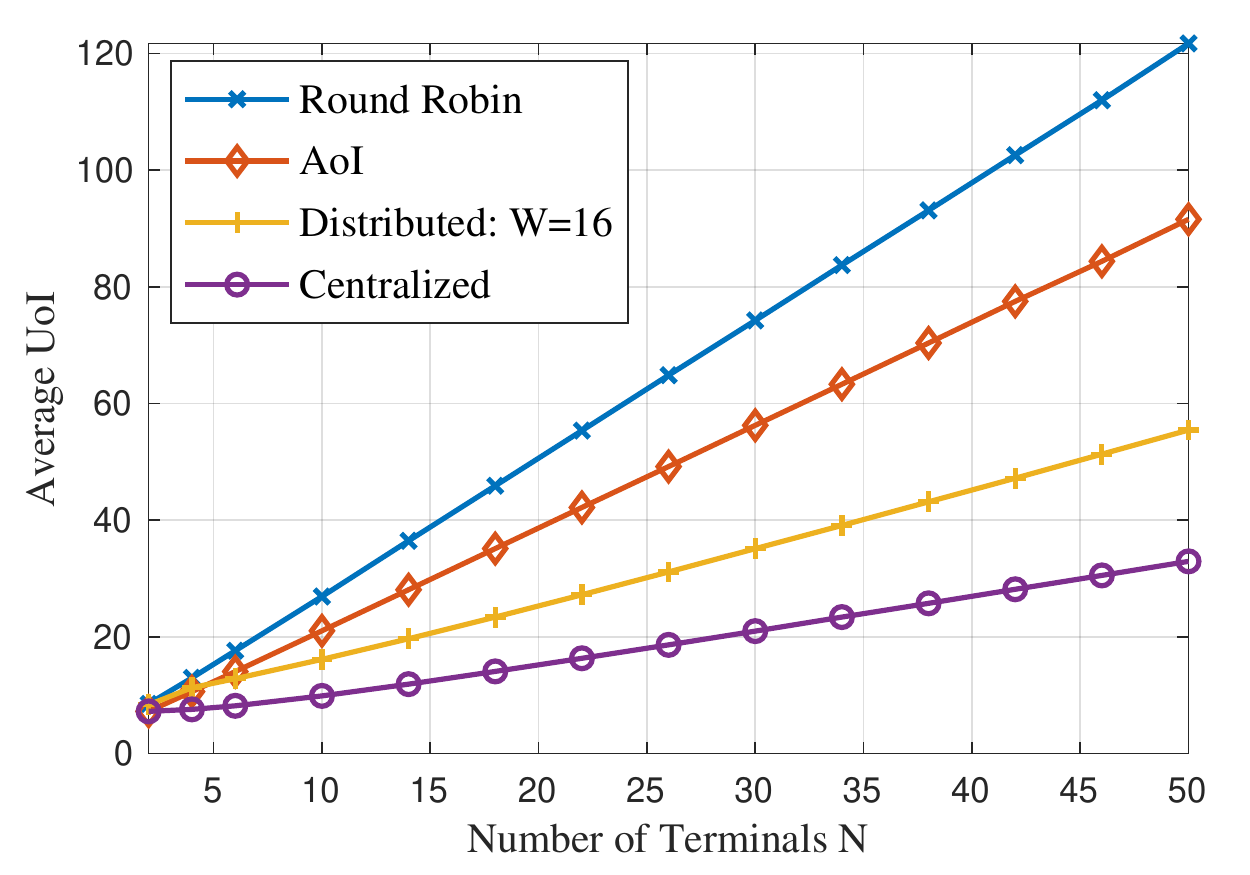}
\vspace*{-.15in}
\caption{Average UoI under several scheduling schemes. }
\vspace*{-.15in}
\label{fig:SAoI}
\end{figure}

\subsubsection{UoI Comparison}
The average UoI under different schemes is illustrated in Fig. \ref{fig:SAoI}. The round robin scheme has the worst performance. With both context and status information, the centralized scheduling scheme produces the lowest average UoI. Overall, both the centralized and the distributed context-aware adaptive scheduling schemes outperform the AoI-based scheme. 

\begin{figure}[htbp]
\centering
\vspace*{-.15in}
\includegraphics[width=3.2in]{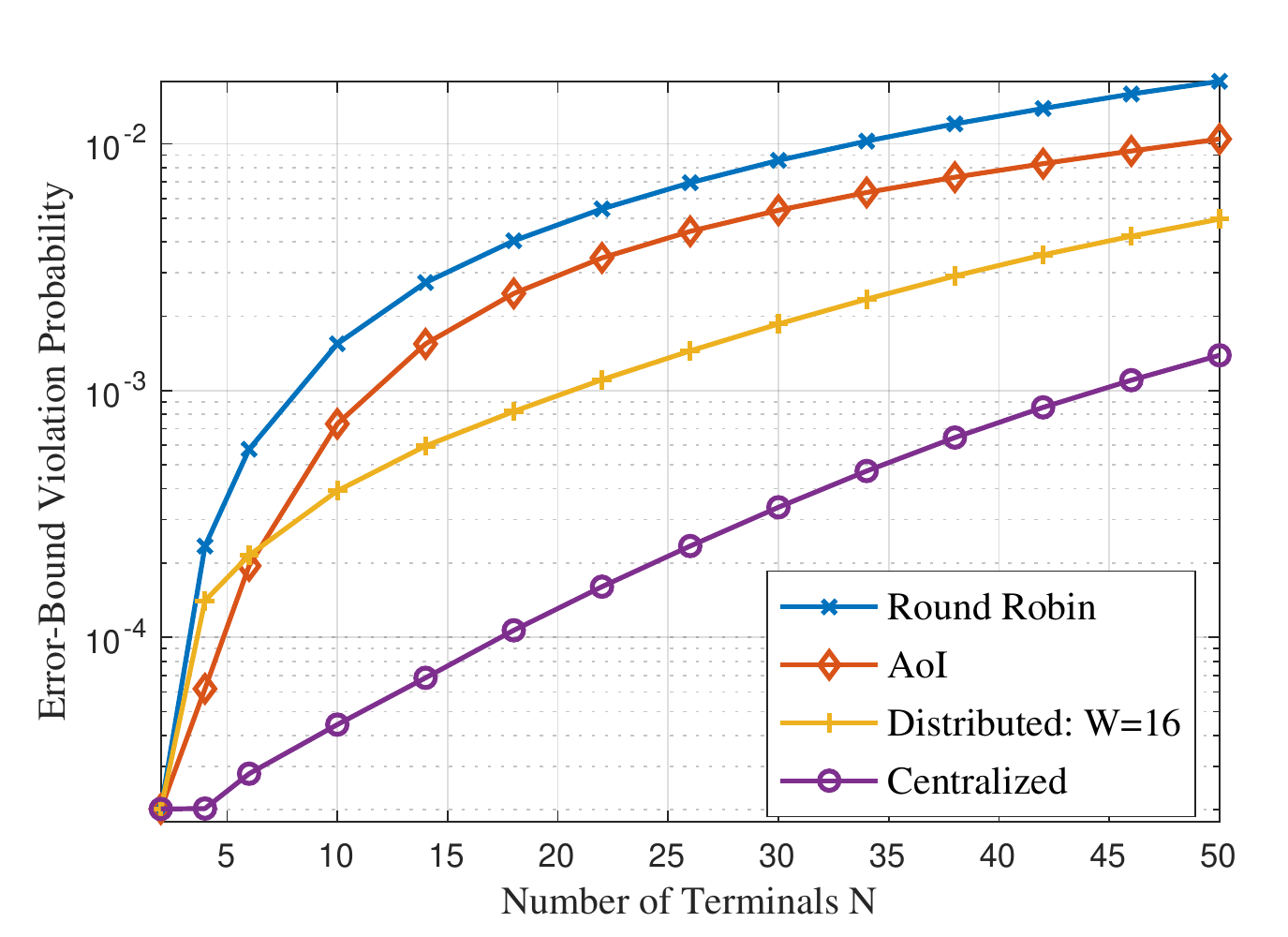}
\vspace*{-.15in}
\caption{Average error-bound violation probability. The threshold for error $Q_{i,t}$ is $15$ when the weight is 1, and 5 when the weight is 100. }
\vspace*{-.15in}
\label{fig:vio}
\end{figure} 

\subsubsection{Error-Bound Violation Probability}
Fig. \ref{fig:vio} plots the average probability of the error $|Q_{i,t}| > 5$ when the context-aware weight is 100 and $|Q_{i,t}| > 15$ when the context-aware weight is 1. It is shown that when the context and status information is known to the scheduler, the centralized scheduling scheme can reduce the error-bound violation probability by around 90\% comparing to the AoI-based scheme. Although the UoI minimization problem is not specifically designed for avoiding error-bound violation, simulation results show that both the centralized and the distributed context-aware adaptive scheduling schemes bring a significant improvement over round robin and AoI-based schemes especially when the number of terminals is large. 

\subsubsection{Remote Control of CartPoles}
We exam the performance of the status update schemes in remotely controlling CartPole systems. In CartPole, a pole is attached to a cart (as shown in Fig. \ref{fig:cartpole}), and a controller forces the cart to its left or right to prevent the pole from falling and the cart from moving out of the screen. The game ends when the angle of the pole is larger than $12^\circ$ or the position of the cart is 2.4 units away from the center, and is played for at most 200 steps in each episode. The controller tries to play as many steps as possible in each episode. In the original CartPole, the 4-dimensional system status (i.e., the position $x$ and velocity $\dot{x}$ of the cart, and the angle $\alpha$ and the angular velocity $\dot{\alpha}$ of the pole) is \emph{fully observable} at each step, such that the controller decides whether to force the cart to its left or its right based on the status. The control problem has been thoroughly studied with many good control algorithms\cite{CartPole}. 

\begin{figure}[htbp]
\vspace*{-.15in}
\centering
\includegraphics[width=2.1in]{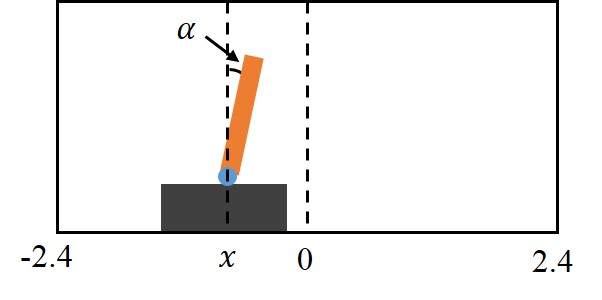}
\vspace*{-.15in}
\caption{An illustration of Cartpole. The controller pushes the cart to the left or right with 10 Newtons of force to keep the pole stand and the cart in the given range. }
\vspace*{-.15in}
\label{fig:cartpole}
\end{figure}

To consider remotely controlling Cartpole systems, an additional random force is applied to the cart at each step, so that the status of the CartPole is not fully predictable to the controller even if the kinetics equations are known. The simulation process is described as follows:
\begin{enumerate}[a)]
\item Train a multi-layer perceptron (MLP) with a 100-node hidden layer for the control problem of a single CartPole whose the status is fully observable. 
\item At each step, schedule the CartPoles to deliver their status to the controller according to a scheduling scheme. For the CartPoles whose status is not delivered, the controller updates their status with previous information and the kinetics equations assuming zero random force. 
\item The controller remotely controls the CartPoles with the control algorithm based on the latest available status. 
\end{enumerate}

\begin{figure}[htbp]
\vspace*{-.15in}
\centering
\includegraphics[width=3.2in]{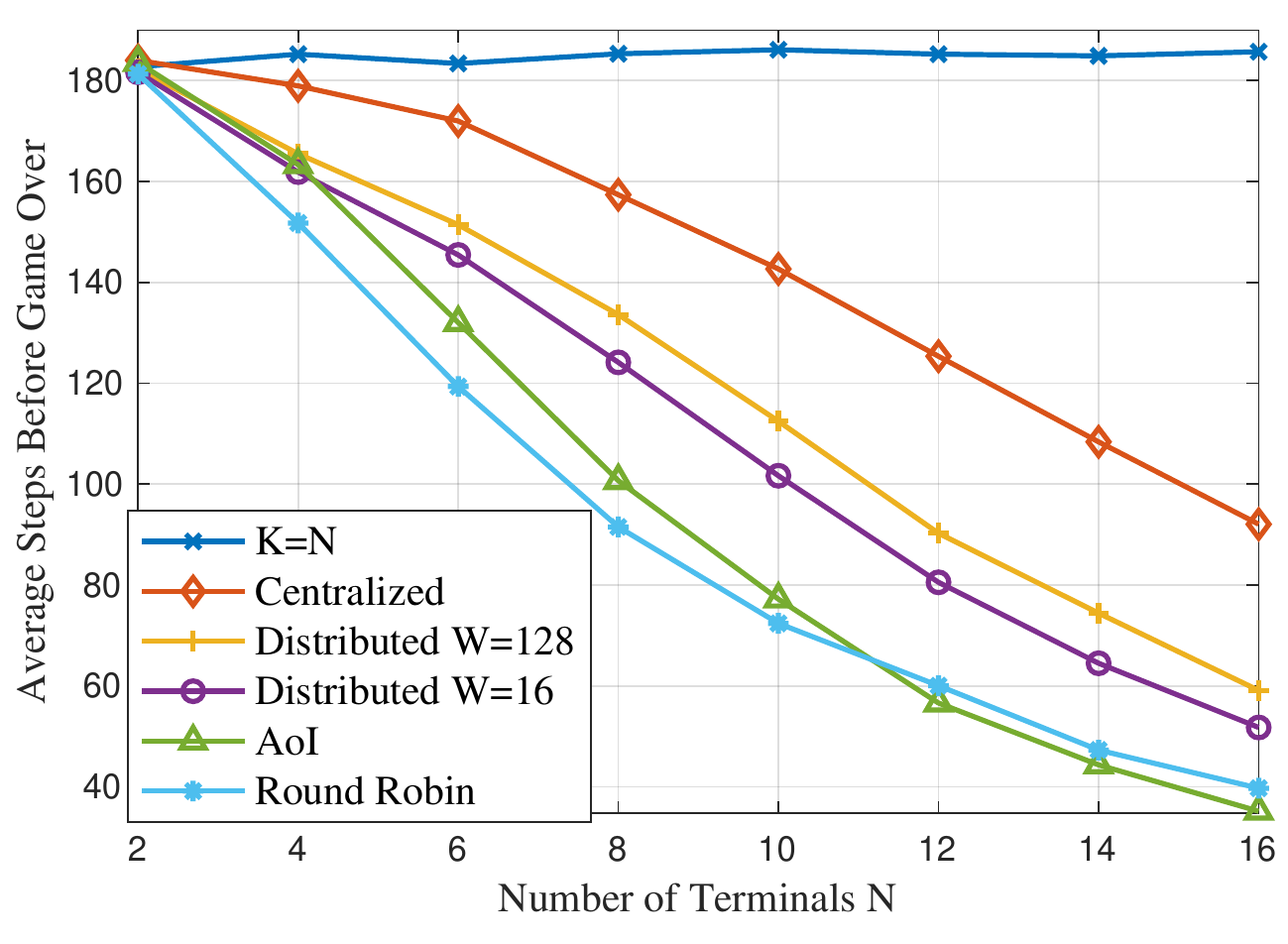}
\vspace*{-.15in}
\caption{Average steps before game over in CartPole. }
\vspace*{-.15in}
\label{fig:ctrl}
\end{figure}

Fig. \ref{fig:ctrl} illustrates the average number of steps being played in each episode of CartPole. The game runs on OpenAI gym\footnote{https://github.com/openai/gym/blob/master/gym/envs/classic\_control/cartpole.py}, which is an open-source platform that provides easy access to many classic control problems. The random force follows Gaussian distribution with zero mean and standard deviation of $10$ Newtons, which equals to the magnitude of the force applied by the controller at each time slot. The context information is exploited under a straightforward intuition: the status information is more urgent if the situation worsens. Therefore, if the distance is going to increase at the next step (i.e., $x\dot{x} > 0$), the context-aware weight for position $x$ and velocity $\dot{x}$ is set to 9; otherwise the context-aware weight is 1. Similarly, the context-aware weight for angle $\alpha$ and angular velocity $\dot{\alpha}$ is 9 if $\alpha\dot{\alpha} > 0$. For each dimension of the status $\left(x,\dot{x},\alpha,\dot{\alpha}\right)$, the UoI is firstly computed. After linearly rescaling the status such that both the value of position and the value of angle are within $[-1,1]$, the UoI of each dimension is summed up as the overall UoI. However, it is hard to implement different lengths of time slots in the infrastructure, so the contention window is neglected in this CartPole simulation. Fig. \ref{fig:ctrl} shows that with UoI-aware scheduling, the CartPoles can be stable for a longer time, meaning that the performance of the CartPole game is substantially improved. 

\section{Conclusions}

This paper proposes UoI as a new metric for the timeliness of status updates in remote control systems. The UoI exploits the context information of the system as well as real-time status evolution. The updating scheme of a single terminal under an average updating frequency constraint is proposed. Then a centralized scheduling scheme is also proposed to allocate limited channel resources to multiple terminals in order to reduce the average UoI. To facilitate decentralized implementation of the proposed scheduling scheme, a dynamic threshold-based random access scheme is designed. Simulation results show that the proposed adaptive updating scheme and scheduling schemes are able to adapt to the context and status information in the system, and achieve a significant improvement over existing AoI-based schemes on the timeliness of status updates, so as to improve the performance of remote control systems with the example of CartPole game. 

\begin{appendices}
\section{Proof of Proposition 1}
The control actions are made to minimize $$\sum_{t=0}^{T-1}\E{\omega_{t}\left(a\hat{x}_{t-1} + bv_t + r_t - y_t\right)^2}.$$
Taking derivative with respect to control action $v_t$, the objective function of problem (\ref{pro:control}) becomes
\begin{eqnarray*}
\frac{\mathrm{d}}{\mathrm{d}v}\sum_{t=0}^{T-1}\E{\omega_{t}\left(a\hat{x}_{t-1} + bv_t + r_t - y_t\right)^2} = 2\E{\omega_{t}}\left(a\hat{x}_{t-1} + bv_t - y_t\right).
\end{eqnarray*}
The last equality holds since $r_t$ has zero mean. Therefore, the minimum weighted squared difference is achieved when 
$v^*_t = \frac{y_t - a\hat{x}_{t-1}}{b}$. 

Therefore, the objective function of problem (\ref{program:LQC}) becomes
\begin{align*}
\lim_{T\to\infty}\frac{1}{T}\sum_{t=0}^{T-1}\E{\omega_{t}\left(x_t - y_t\right)^2}
&= \lim_{T\to\infty}\frac{1}{T}\sum_{t=0}^{T-1}\E{\omega_{t}\left(ax_{t-1} + bv^*_t + r_t  - y_t\right)^2}\\
&=a^2\lim_{T\to\infty}\frac{1}{T}\sum_{t=0}^{T-1}\E{\omega_{t}\left(x_{t-1} - \hat{x}_{t-1}\right)^2} + \lim_{T\to\infty}\frac{1}{T}\sum_{t=0}^{T-1}\E{\omega_{t}}\E{r^2}.
\end{align*}
Thus, the objective is equivalent to minimizing the long-term average weighted squared difference between the actual state and the estimated state. 

\section{Proof for Lemma 1}
By the definition of virtual queue $H_t$ in Eq. (\ref{eqn:vq}),
\begin{eqnarray}
\label{eqn:qstable}
H_T~\geq~H_{T-1} - \rho + U_{T-1}~\geq~H_0 - T\rho + \sum_{\tau=0}^{T-1}U_\tau.
\end{eqnarray}
Dividing both sides by $T$ and taking limit to infinity yields
\begin{eqnarray*}
\lim_{T\to\infty}\frac{\E{H_T} - H_0}{T} &\geq& - \rho + \lim_{T\to\infty}\frac{1}{T}\sum_{\tau=0}^{T-1} \E{U_\tau}.
\end{eqnarray*}
Since the left-hand side equals to zero, the lemma is hereby proved. 

\section{Proof for Lemma 2}
\label{sec:lemma2}
Summing up Eq. (\ref{eqn:optimization-single}) over $t\in\{0,1,2\cdots,T-1\}$, we get
\begin{eqnarray}
\label{proof:3-1}
\E{L_T} - \E{L_0} + \sum_{t=0}^{T-1}\E{f_t}~\leq~CT.
\end{eqnarray}
First, we prove that the virtual queue is mean rate stable. By $\E{f_t}\geq \bar f_\mathrm{min}$, we get
$$\E{L_T}~\leq~(C-\bar f_\mathrm{min})T + \E{L_0}.$$
By the definition of Lynapunov function, we have 
$\frac{1}{2}V\E{H^2_T} \leq (C-\bar f_\mathrm{min})T + \E{L_0}$.
Since $\E{H^2_T}\geq\E{H_T}^2$, we obtain 
$\E{H_T}~\leq~\sqrt{\frac{2}{V}\left((C-\bar f_\mathrm{min})T + \E{L_0}\right)}$. 
Dividing both sides of the inequality by $T$ and letting $T$ approach $\infty$ yields
$\limsup_{T\to\infty}\frac{\E{H_T}}{T}=0$.

Next, we prove the upper bound in (\ref{eqn:bound-single}). Dividing (\ref{proof:3-1}) by $T$ and letting $T\to\infty$ yields
\begin{eqnarray*}
\limsup_{T\to\infty}\frac{1}{T}\E{L_T} + \limsup_{T\to\infty}\frac{1}{T}\sum_{t=0}^{T-1}\E{f_t}~\leq~C.
\end{eqnarray*}
By $L_t\geq0$, the lemma is proved. 

\section{Proof for Lemma 3}
\label{sec:lemma3}
By Eq. (\ref{con:dynamic}) and $U_t,S_t\in\left\{0,1\right\}$, we have 
\begin{eqnarray}
Q_{t+1}^2-Q_{t}^2 & = & A_t^2 + 2\left(1-U_tS_t\right)A_tQ_t - U_tS_tQ_t^2. \label{eqn:drift-q}
\end{eqnarray}
By the definition of virtual queue in Eq. (\ref{eqn:vq}), we get
\begin{eqnarray}
H_{t+1}^2-H_{t}^2 \leq \left(H_t - \rho + U_t\right)^2 - H_t^2 \leq  1 + 2\left( - \rho + U_t\right)H_t. \label{eqn:drift-h}
\end{eqnarray}
Substituting Eq. (\ref{eqn:drift-q}) and Eq. (\ref{eqn:drift-h}) into Eq. (\ref{def:drift}) yields
\begin{eqnarray*}
\Delta_t &\leq& \frac{1}{2}V + V\left( - \rho + \E{U_t|Q_t, \omega_{t+1}, H_t}\right)H_t + \theta \sigma^2  - \theta p\E{U_t|Q_t, \omega_{t+1}, H_t}Q_t^2\\
&~& + 2\theta\E{\left(1-U_tS_t\right)|Q_t, \omega_{t+1}, H_t}\E{A_t|Q_t, \omega_{t+1}, H_t}Q_t. 
\end{eqnarray*}
Since $\E{A_t}=0$ and by the assumption that $A_t$ is independent to $Q_t, \omega_{t+1}, H_t$, we have
\begin{eqnarray}
\label{ineqn:proof1}
\Delta_t \leq \theta\sigma^2 + \frac{1}{2}V - V\rho H_t + (VH_t - \theta pQ_t^2)\E{U_t|Q_t, \omega_{t+1}, H_t}. 
\end{eqnarray}
By adding the expected penalty $\E{f_t|Q_t, \omega_{t+1}, H_t}$ to both sides of Eq. (\ref{ineqn:proof1}), the lemma is proved. 

\section{Proof for Lemma 4}
\label{sec:lemma4}
Given the dynamic function in Eq. (\ref{eqn:dynamic}) and $U_{i,t}\in\left\{0,1\right\}$, $S_{i,t}\in\left\{0,1\right\}$, we have
\begin{eqnarray*}
Q_{i,t+1}^2-Q_{i,t}^2 & = & A_
{i,t}^2 + 2\left(1-U_{i,t}S_{i,t}\right)A_{i,t}Q_{i,t} - U_{i,t}S_{i,t}Q_{i,t}^2. 
\end{eqnarray*}
Substitute it into Eq. (\ref{def:drift-multi}). Using $A_{i,t}$ is independent to $\boldsymbol{Q}_t, \boldsymbol{\omega}_{t+1}$ and has zero mean yields
\begin{eqnarray*}
\Delta_t
=\sum_{i=1}^N\theta_i\sigma_i^2 - \sum_{i=1}^N\theta_ip_iQ_{i,t}^2\E{U_{i,t}|\boldsymbol{Q}_t, \boldsymbol{\omega}_{t+1}}.
\end{eqnarray*}
The expectation is exactly
$$\E{f_t|\boldsymbol{Q}_t, \boldsymbol{\omega}_{t+1}} = \sum_{i=1}^N\omega_{i,t+1}\sigma_i^2 + \sum_{i=1}^N\omega_{i,t+1}Q_{i,t}^2\left(1-p_i\E{U_{i,t}|\boldsymbol{Q}_t, \boldsymbol{\omega}_{t+1}}\right).$$
Summing up the above two equations, the lemma is hereby proved. 

\section{Proof for Theorem 2}
\label{sec:the2}
Letting $\boldsymbol{\pi} = \left(\pi_1,\pi_2,\cdots,\pi_N\right)$ be a stationary scheme of program (\ref{program:main}), with which the $i$-th terminal is scheduled with probability $\pi_i$ at each time slot. Since scheme (\ref{policy:0-multi}) minimizes the right-hand side of Eq. (\ref{eqn:driftpluspenalty-multi}), substituting scheme $\boldsymbol{\pi}$ into the right-hand side of Eq. (\ref{eqn:driftpluspenalty-multi}), we get
\begin{eqnarray*}
\E{L_{t+1} - L_t + f_t|\boldsymbol{Q}_t, \boldsymbol{\omega}_{t+1}}
&\leq&  \sum_{i=1}^N(\theta_i+\omega_{i,t+1})\sigma_i^2 - \sum_{i=1}^N\theta_ip_i\pi_iQ_{i,t}^2 + \sum_{i=1}^N\omega_{i,t+1}\left(1-p_i\pi_i\right)Q_{i,t}^2. 
\end{eqnarray*}
Since context-aware weight $\boldsymbol{\omega}_{t+1}$ is independent to estimation error $\boldsymbol{Q}_t$, taking expectation yields
\begin{eqnarray}
\label{eqn:drift-bound-multi}
\E{L_{t+1} - L_t + f_t|\boldsymbol{Q}_t}
&\leq&  \sum_{i=1}^N(\theta_i+\bar\omega_i)\sigma_i^2 + \sum_{i=1}^N\left(\bar{\omega}_i\left(1-p_i\pi_i\right)-\theta_ip_i\pi_i\right)Q_{i,t}^2.
\end{eqnarray}
If $\theta_i \geq \frac{\bar{\omega}_i(1-p_i\pi_i)}{p_i\pi_i}$, the last term at the right-hand side of Eq. (\ref{eqn:drift-bound-multi}) is no larger than zero. To minimize the first term at the right-hand side of Eq. (\ref{eqn:drift-bound-multi}), letting $\theta_i = \frac{\bar{\omega}_i(1-p_i\pi_i)}{p_i\pi_i}$, we get
\begin{eqnarray}
\label{eqn:lya-multi}
\E{L_{t+1} - L_t + f_t|\boldsymbol{Q}_t} \leq  \sum_{i=1}^N\frac{\bar{\omega}_i\sigma_i^2}{p_i\pi_i}, 
\end{eqnarray}
in which the right-hand side is a constant. Taking expectation at both sides of Eq. (\ref{eqn:lya-multi}), we have
\begin{eqnarray*}
\E{L_{t+1} - L_t + f_t} \leq  \sum_{i=1}^N\frac{\bar{\omega}_i\sigma_i^2}{p_i\pi_i}.
\end{eqnarray*}
Summing up the above inequality over $t\in\{0,1,\cdots,T-1\}$, and dividing both sides by $T$. Taking limit as $T\to\infty$, since $L_0<\infty$ and $L_t\geq0$, we obtain
\begin{eqnarray*}
\limsup_{T\to\infty}\frac{1}{T}\sum_{t=0}^{T-1}\E{f_t}\leq\sum_{i=1}^N\frac{\bar{\omega}_i\sigma_i^2}{p_i\pi_i}. 
\end{eqnarray*}
Since $f_t$ is the UoI at the $(t+1)$-th time slot, the theorem is proved hereby. 

\section{Proof for Lemma 5}
Without loss of generality, let the first to the $K$-th terminal be active. Thus, we have 
$$W_K = \max_{i\in\{1,2,\cdots,K\}}{l_i} + 1.$$
Since each of the $K$ terminals occupies a sub-channel, the probability distribution of $W_K$ is 
\begin{eqnarray*}
\Prob{W_K\leq t} & = &\Prob{l_i<t, \forall i\in\{1,2,\cdots,K\}|l_1\neq l_2\neq \cdots \neq l_K}, 
\end{eqnarray*}
which leads to
\begin{eqnarray*}
\Prob{W_K\leq t} =\left\{
\begin{aligned}
&\frac{\binom{t}{K}}{\binom{W}{K}}, &t \geq K;\\
&0, &t < K. 
\end{aligned}
\right.
\end{eqnarray*}
So its expectation is obtained by
\begin{eqnarray*}
\E{W_K} = W - \sum_{t=K}^{W-1}\Prob{W_K \leq t} =  W - \frac{\sum_{t=K}^{W-1}\binom{t}{K}}{\binom{W}{K}}. 
\end{eqnarray*}
Substituting $\sum_{t=K}^{W-1}\binom{t}{K} = \binom{W}{K+1}$ into the above equation yields
$\E{W_K} = \frac{K}{K+1}(W+1)$. 

\end{appendices}

\bibliographystyle{ieeetr}

\end{document}